\documentclass[10pt, compsocconf, fleqn]{IEEEtran}
\ifCLASSOPTIONcompsoc
\usepackage[nocompress]{cite}
\else
\usepackage{cite}
\fi
\usepackage{bbm}
\usepackage{graphicx}
\usepackage{cite}
\usepackage{subcaption}
\usepackage{algorithm}
\usepackage{amsmath,amsthm,amssymb}
\usepackage{multirow}
\usepackage{epstopdf}
\usepackage{color}
\usepackage{epsfig}
\usepackage{xfrac}
\usepackage{filecontents}
\usepackage{csquotes}
\usepackage{footnote}
\makesavenoteenv{tabular}
\makesavenoteenv{table}
\usepackage{url}
\usepackage{balance}
\newtheorem{theorem}{Theorem}

\newtheorem{remark}[theorem]{Remark}
\theoremstyle{definition}
\newtheorem{definition}{Definition}[section]
\usepackage{pifont}
\usepackage{xfrac}
\usepackage{pbox}
\usepackage{epstopdf}
\usepackage[flushleft]{threeparttable}

\makeatletter
\def\Let@{\def\\{\notag\math@cr}}
\makeatother

\ifCLASSINFOpdf
\else
\fi

\begin{document}
	%
	\title{Securing Organization's Data: A Role-Based Authorized Keyword Search Scheme with Efficient Decryption}
	\author{Nazatul Haque~Sultan,~\IEEEmembership{Student Member,~IEEE,}
		Maryline~Laurent,~\IEEEmembership{Member,~IEEE,}
		and Vijay~Varadharajan,~\IEEEmembership{Senior Member,~IEEE}
		\IEEEcompsocitemizethanks{\IEEEcompsocthanksitem N. H. Sultan and M. Laurent are with the RST Department, T\'el\'ecom SudParis, Institut Polytechnique de Paris, France.\protect\\
			E-mail: nazatulhaque.sultan@gmail.com; maryline.laurent@telecom-sudparis.eu
			\IEEEcompsocthanksitem V. Varadharajan is with the Faculty of Engineering and Built Environment, Global Innovation Chair Professor at The University of Newcastle, Callaghan, Australia.\protect\\
			E-mail: Vijay.Varadharajan@newcastle.edu.au
		}
	}

	\maketitle

	%
	
	\IEEEdisplaynontitleabstractindextext
	
	\IEEEpeerreviewmaketitle

	\begin{abstract}
		For better data availability and accessibility while ensuring data secrecy, organizations often tend to outsource their encrypted data to the cloud storage servers, thus bringing the challenge of keyword search over encrypted data. In this paper, we propose a novel authorized keyword search scheme using Role-Based Encryption (RBE) technique in a cloud environment. The contributions of this paper are multi-fold. First, it presents a keyword search scheme which enables only the authorized users, having proper assigned roles, to delegate keyword-based data search capabilities over encrypted data to the cloud providers without disclosing any sensitive information. Second, it supports a multi-organization cloud environment, where the users can be associated with more than one organization. Third, the proposed scheme provides efficient decryption, conjunctive keyword search and revocation mechanisms. Fourth, the proposed scheme outsources expensive cryptographic operations in decryption to the cloud in a secure manner. Fifth, we have provided a formal security analysis to prove that the proposed scheme is semantically secure against Chosen Plaintext and Chosen Keyword Attacks. Finally, our performance analysis shows that the proposed scheme is suitable for practical applications.
	\end{abstract}
	
	\section{Introduction}
	\label{introduction}
	With the ever-increasing amount of digital information, individuals and organizations are now storing/outsourcing their data in the cloud to make use of features such as better accessibility, high availability, reduction of maintenance and initial investment costs \cite{Ferrer2019}. However, with sensitive data stored in the cloud (e.g. see McAfee report \cite{McAfee}) and legal concerns (such as compliance to the European General Data Protection Regulation - GDPR\footnote{\url{https://ec.europa.eu/commission/priorities/justice-and-fundamental-rights/data-protection/}}), security and privacy have become major issues in cloud data storage\footnote{In this paper, cloud represents the public cloud that provides storage facilities to the general public (i.e., individuals and organizations). In general, the public cloud is maintained by a third-party entity referred to as \emph{Cloud Service Provider} \cite{NIST}.}. To preserve privacy and confidentiality of outsourced data in the cloud, a preferred technique that is often used is \emph{encryption-before-outsourcing}. The encryption-before-outsourcing technique enables the data owners (i.e. entities owning the data) to outsource their sensitive data in the cloud in an encrypted form. As such, no entity including the cloud service provider can access the sensitive plaintext data without having access to proper decryption key. This, however, restricts data retrieval/search over encrypted data \cite{Bosch2014}. A trivial solution is to download the whole encrypted database, and then perform the search operation locally after decryption. It is clear that this is not practical. An alternative approach is to allow the service provider to decrypt all the encrypted data so that it can perform search operation over the plaintext data. However, this violates data privacy. 
	\par 
	
	Searchable Encryption (SE) has gained a considerable amount of interest from the research community to address the issue of searching over encrypted data \cite{HAN201666}. In SE, users delegate data search capabilities for some keywords over the encrypted data to a service provider without disclosing any useful information about the searched keywords and the actual content of the encrypted data. This process is also referred to as \emph{keyword search}. Typically, in keyword search, data owners outsource their data in an encrypted form along with an encrypted index of keywords. Whenever a user wants to access data, the user sends the desired keywords in the form of trapdoors to the service provider. In return, the service provider uses the trapdoors to perform search over the encrypted indexes and sends the associated encrypted data, if there is a match between the keywords associated with the trapdoor and encrypted indexes. 
	\par 
	Many works have been done in the area of keyword search, achieving search authorization in a coarse-grained way. That is, the users can search all the keywords using their secret keys \cite{Hu2017}. However, this kind of authorization may disclose sensitive information. For example, Organization A outsources its data files to the cloud so that its employees can easily access them. Assume Organization A is a participant in a consortium with another Organization B and other organizations. Suppose, some files are associated with the keywords ``Organization B" and ``Project X" which are only allowed to be
	accessed by the Managers in the Organization A. In this case, if an adversary can search for the keywords ``Organization B" and ``Project X" and gets all the encrypted files associated with these two keywords. This will eventually reveal, without knowing the actual content, that Organization A and Organization B are collaborating on Project X, which may not be desirable.
	
	\par 
	To address this problem, several authorized keyword search schemes have been proposed for \emph{multi-user settings} using different cryptographic techniques, e.g. Pairing-Based Encryption \cite{Bao2008}, Predicate Encryption \cite{Li2011} and Attribute-Based Encryption \cite{Sun2016, Hu2017}, where multiple users are able to perform keyword search operations based on some access policies. However, none of these techniques efficiently support hierarchies in an organization, where higher level authorities can inherit access rights of their subordinates. As such, all these schemes \cite{Bao2008, Li2011, Sun2016, Hu2017} are not able to reflect efficiently organization's policies and structures\footnote{In an organization, typically employees are organized in a hierarchical way based on their responsibilities and qualification \cite{Zhou2013}.} \cite{Perez2017}.
	
	\par 
	Role-Based Encryption (RBE) \cite{Zhou2011, Zhou2013, Zhu2013} is an emerging cryptographic technique, which combines both properties of the traditional Role-Based Access Control (RBAC) \cite{Sandhu1996} and cryptographic encryption methods, to achieve data access control over encrypted data. In RBE, the data owner encrypts data using a RBAC access policy defined over some roles\footnote{In an organization, roles are typically created based on job functions.}, and any user having proper roles can derive the secret keys for decryption. Unlike the traditional RBAC method, RBE enables the data owners to define and enforce RBAC access policies on the encrypted data itself. This, in turn, reduces the dependency of the data owners on untrusted service provider for defining and enforcing access policies while sharing data with other authorized users. Moreover, similar to the RBAC, in RBE, roles can inherit access permissions from other roles \cite{Zhou2013}. Hence, the roles can be organized in a hierarchical structure. This is one of the main advantages of RBE over other encryption mechanisms such as Attribute-Based Encryption \cite{Bethencourt2007, Goyal2006}, as it can reflect closely a real-world organisation's policies and structure. The inheritance property of the RBE makes it more suitable for large scale organizations such as enterprises with a complex hierarchical structures \cite{Zhou2013}. Therefore, RBE is a more suitable cryptographic technique for designing a keyword search mechanism compared with other cryptographic techniques such as the ABE.
	\par 
	RBE has been used to provide data access control in cloud environments over encrypted data \cite{Zhou2013, Zhu2013, Perez2017}. However, they mainly focus on a single organization cloud environment scenario, where users can have roles only in a single organization and hence can access data associated with only that organization. In many practical scenarios in a cloud environment, a data owner may want to share his/her data with users in several organizations having different roles. For example, a user may work as a researcher and doctor in a clinical research laboratory and hospital respectively. As such, the same user will hold roles in the clinical research laboratory and the hospital. The data owner can specify a RBAC access policy in such a way that only the users having the access privileges for the roles ``Researcher" and ``Doctor" can gain access to the actual content corresponding to the encrypted data.

	\par 
	This paper further investigates the aforementioned research gaps and proposes a novel keyword search scheme using the RBE technique where organizations outsource their data to a public cloud. The proposed scheme supports a multi-organization environment, where users can possess roles from more than one organization. It also enables the data owners to define and enforce RBAC access policies on encrypted data, thereby allowing any a user having authorized roles to perform a keyword search along with the ability to decrypt. 
	The salient features of the proposed scheme are as follows:
	\begin{enumerate}
		\item An authorized keyword search mechanism is proposed using RBE technique so that only the users possessing authorized roles can delegate keyword search capabilities over encrypted data to the public cloud. 
		\item The proposed scheme supports multi-organization cloud environment, where a user can be associated with more than one organization, having one or more roles in different organizations.
		\item Conjunctive keyword search\footnote{In conjunctive keyword search, a user can search for multiple keywords in a single request  \cite{Ferrer2019}.} functionality is supported without any significant overhead in the system.

		\item A user revocation mechanism has been introduced to revoke unintended users.
		\item An outsourced decryption mechanism is combined with the proposed scheme enabling the users to delegate most of the computationally expensive cryptographic operations to the public cloud, thereby reducing the overhead on the user-side.
		\item A formal security analysis of the proposed scheme has been given demonstrating that the scheme is secure against the Chosen Plaintext Attacks and the Chosen Keyword Attacks.
		\item  A performance analysis of the proposed
		scheme has been provided which shows that the proposed scheme is sufficiently efficient to be used in practical applications.
	\end{enumerate}
	
	\par 
	The organization of this paper is as follows: Section \ref{related_works} presents a brief overview of some existing works related to the proposed scheme. Section \ref{problem_statement} outlines the problem statement, where the system model, threat model, design and security goals, frameworks and security model of the proposed scheme are presented. Section \ref{preli} gives a brief overview of the role hierarchy, bilinear pairing properties, a group key distribution technique, and some mathematical assumptions, which will be used throughout this paper. Section \ref{proposed_scheme} details the proposed scheme including an overview followed by its main construction. Section \ref{analysis} presents a detailed security and performance analyses of the proposed scheme, and finally section \ref{conclusion} concludes this paper.
	
	\section{Related Works}
	\label{related_works}
	
	This section presents a brief overview of some notable works in the keyword search area, including some cryptographic RBAC based data access control schemes.
	
	\subsection{Keyword Search over Encrypted Data}
	Data search over encrypted data has been extensively studied since the past decade. Song \emph{et al.} presented the first practical symmetric key cryptography based searchable encryption scheme that can search full text over encrypted data \cite{Song2000}. Leter, several searchable encryption schemes have been proposed, for various functionalities and security requirements, based on either symmetric key cryptography (SKC) \cite{Curtmola2006, Kamara2012, Li2019, Hoang2019, Liu2018} or public-key cryptography (PKC) \cite{Boneh2004, Boneh2007, Sun2016, Hu2017, Miao2017, Chaudhari2019}. 
	\par 
	In \cite{Curtmola2006}, Curtmola \emph{et al.} proposed a SKC based keyword search scheme for \emph{multi-user} settings\footnote{Multi-user settings enable the data owners to authorize any number of users to perform keyword search operations.}, which can perform single keyword search. In \cite{Kamara2012}, Kamara \emph{et al.} proposed a dynamic version of the scheme \cite{Curtmola2006} that can add and delete files at any time efficiently. However, the scheme \cite{Kamara2012} leaks significant information while performing update operation \cite{Hoang2019}. In \cite{Li2019}, Li \emph{et al.} proposed a SKC based forward search privacy scheme, which prevents any leakage of information about the past queries. Later on, in \cite{Liu2018}, Liu \emph{et al.} proposed a keyword search scheme which enables the users to verify the search results against the dishonest servers. Although the SKC based keyword search schemes provides better efficiency in terms computation cost, PKC based keyword search schemes provide more flexible and expressive search queries \cite{Sun2016}. 
	\par 
	Recently, many PKC based authorized keyword search schemes have been proposed based on Attribute-Based Encryption (ABE) \cite{Sun2016, Hu2017, Miao2017, Chaudhari2019, Sultan2019}, where any user having a qualified set of attributes that satisfy an access policy can perform search operation using some keywords. That is, these schemes provide authorized keyword search, which allows only intended users to do the search in multi-user settings. In \cite{Sun2016, Sultan2019}, Sun \emph{et al.} and Sultan \emph{et al.} proposed keyword search schemes using ABE technique. The schemes provide both single and conjunctive keyword search without introducing any additional overhead in the system. In \cite{Hu2017}, Hu \emph{et al.} proposed another ABE based keyword search scheme for dynamic policy update, where the data owners can securely update the access policies using proxy re-encryption and secret sharing techniques. In \cite{Miao2017}, Miao \emph{et al.} proposed an ABE based keyword search scheme for hierarchical data, which also supports conjunctive keyword search. In \cite{Chaudhari2019}, Chaudhari \emph{et al.} proposed an authorized keyword search scheme using ABE, which hides the access policy from all the intended entities including the public cloud. However, all the aforementioned schemes do not support role hierarchy property and inheritance property. 

	
	
	\subsection{Cryptographic RBAC based Data Access Control}
	A cryptographic RBAC based data access control mechanism integrates the traditional RBAC model with cryptographic encryption method to enforce RBAC access policy on encrypted data. It enables the data to be encrypted using RBAC access policy defined over some role(s). Any user, possessing the required role(s) satisfying the associated RBAC access policy is allowed to decrypt the data. Some notable works in this area are \cite{Akl1983, Lin2011, Tang2016b, Chen2017, Pareek2018, Zhou2013, Zhu2013, Perez2017}, where \cite{Akl1983, Lin2011, Tang2016b, Chen2017, Pareek2018} are based on Hierarchical Key Assignment (HKA) method and \cite{Zhou2013, Zhu2013, Perez2017} are based on RBE method. 
	\par 
	Access control using HKA method has been studied in the early 1980s. In \cite{Akl1983}, Akl \emph{et al.} presented the first cryptographic hierarchical access control technique to solve the hierarchical multi-level security problem, where authorized users are allowed to possess different access privileges. The users are grouped into disjoint sets (or classes) and form a hierarchical structure of classes. Each class is assigned with a unique encryption key and a public parameter in such a way that a higher-level class can derive encryption keys of any lower-level classes using its own encryption key and some public parameters. Later on, several other hierarchical access control schemes have been proposed using different techniques, e.g. \cite{Lin2011, Tang2016b, Chen2017, Pareek2018}. However, the main drawback of the HKA schemes is the high complexity in setting up the encryption keys for a large set of users \cite{Zhou2013}. Also, the user revocation is a challenging task, as all the encryption keys that are known to the revoked users, and their related public parameters need to update per user revocation which may incur a high overhead on the system.
	\par 
	In \cite{Zhou2013}, Zhou \emph{et al.} proposed the first RBE scheme for data sharing in an untrusted hybrid cloud environment. In \cite{Zhou2013}, the ciphertexts and secret keys of the users are constant in size. This scheme also offers user revocation capability. In  \cite{Zhu2013}, Zhu \emph{et al.} proposed another RBE scheme. In this scheme, the ciphertext size linearly increases with the number of roles. In \cite{Perez2017}, Perez \emph{et al.} proposed a data-centric RBAC based data access control mechanism for cloud storage systems using the concept of proxy re-encryption and identity-based encryption techniques. To share data with the authorized users, the data owner generates proxy re-encryption keys based on some RBAC access policies and keeps the re-encryption keys along with the ciphertexts in the cloud storage servers. When an authorized user accesses the ciphertext, the service provider re-encrypts the ciphertext using the proxy re-encryption keys based on a RBAC access policy. However, none of \cite{Zhou2013, Zhu2013, Perez2017} support multi-organization cloud storage systems, where the same user can possess roles from more than one independent organizations. Moreover, \cite{Zhou2013, Zhu2013, Perez2017} do not address keyword search functionality.
	
	\begin{figure}[t]
		\centering
		\scalebox{3.5}{\includegraphics[width=2.6cm, height=1.7cm]{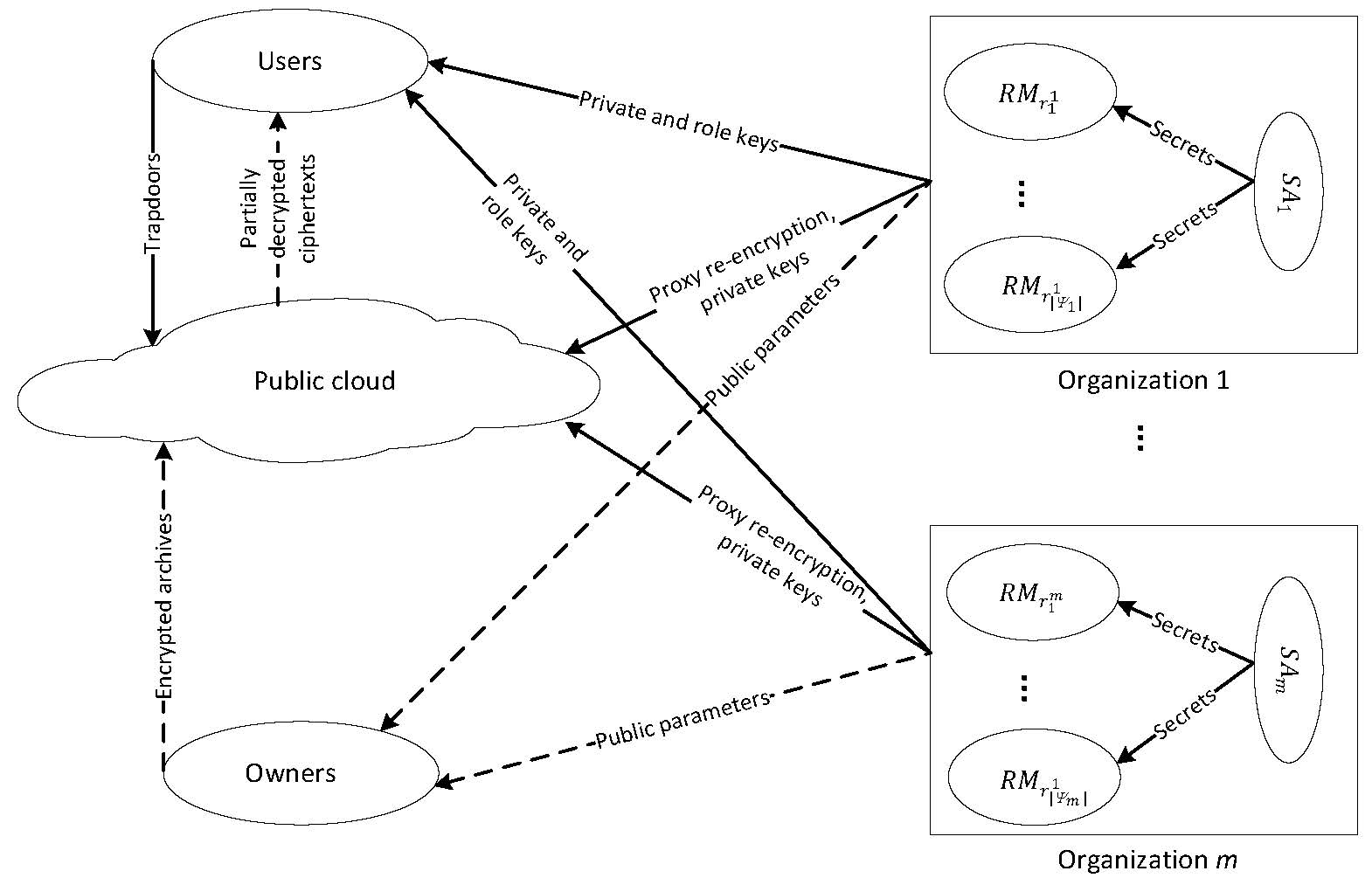}}
		\caption{Proposed System Model}
		\label{system_model}
	\end{figure}
	
	\section{Problem Statement}
	\label{problem_statement}
	This section presents the \emph{System Model}, \emph{Threat Model}, \emph{Design and Security Goals}, \emph{Framework}, and \emph{Security Models} of the proposed scheme.  
	\subsection{System Model}
	\label{cloud_model}
	Figure \ref{system_model} shows the proposed system model, where the doted and dark lines represent public channel and secure-channel such as SSL (Secure Sockets Layer) respectively. It comprises five entities, namely, \emph{System Authorities}, \emph{Role-Managers}, \emph{Data Owners}, \emph{Users}, and \emph{Public Cloud} having the following responsibilities:  
	\begin{itemize}
		\item \emph{System Authority (SA)}: Each organization has one SA, which maintains the role hierarchy of that organization. It generates system public parameters and master secrets for the organization. SA also maintains all the role-managers that are associated with the organization, and it issues  secret keys for each role-manager. In addition, SA issues private and public keys for all the registered users. Further, it issues private, public and proxy re-encryption keys to the public cloud. Moreover, SA is responsible for revoking users from the system when needed.
		
		\item \emph{Role-Manager (RM)}: It is an entity of an organization which manages the role(s). Note that, the roles are assigned by the SA. In addition, it also issues and manages role-keys for the users. 
		
		\item \emph{Data Owners (owners)}: It is an entity who owns the data and wants to outsource his/her data to the public cloud. An owner first encrypts data using a RBAC access policy before outsourcing to the public cloud. The owner first encrypts a plaintext data using a random secret key by following any secure symmetric key encryption algorithm, e.g., \emph{Advanced Encryption Standard} (AES). Afterward, the owner chooses a set of keywords associated with the plaintext data and encrypts those keywords along with the random key using the chosen RBAC access policy. The owner then combines all the ciphertexts into one archive and outsources it to the cloud storage servers.
		
		\item \emph{Users}: It is an entity who wants to access the outsourced data. Each user must register with SA(s) to receive private and public keys associated with the organization(s) from which he/she wants to access data. Also, a user receives a unique role-key for each role he/she possesses from the respective role-manager. When a user wants to access data, the user computes a trapdoor using his/her private keys, role-keys, the desired keyword(s) and sends it to the public cloud. 
		
		\item \emph{Public Cloud}: It is a third-party entity which manages the cloud storage servers. The main responsibility of the public cloud is to store owners' encrypted data. Moreover, it is also responsible for performing keyword search operation over the encrypted data. It is assumed that the public cloud correctly performs search operations using the received trapdoors if and only if the requested user has proper roles. It is also assumed that it partially decrypts all the ciphertexts that have a matching keyword(s) with the trapdoors.
	\end{itemize}
	\subsection{Threat Model}
	\label{threat_model}
	Public cloud is considered as an honest-but-curious entity. That is, public cloud honestly performs all the assigned tasks, but it may try to gain additional privacy information from the data available to it. The users may be malicious, and they may try to collude among themselves to gain access to the data beyond their access privileges. The users, having insufficient access rights, may also try to collude with the public cloud for gaining access to the data beyond their access rights. It is assumed that all the SAs and RMs are fully trusted entities. 
	The threat model is supplemented by a Security Model in Section \ref{security_model}.
	
	\subsection{Design and Security Goals}
	\label{design_security_goals}
	The proposed scheme aims to achieve the following functionality and security goals.
	\par 
	\textbf{Functionality Goals}: The proposed scheme should provide the following functionalities.
	\begin{enumerate}
		\item \emph{Authorized Keyword Search}: Only the users, having proper roles according to the defined RBAC access policy, are authorized to perform keyword search operations over the encrypted data. That is, any unintended users should not get access to the encrypted (outsourced) data. 
		\item \emph{Role-Based Data Sharing}: Only the users, possessing the proper roles according to the defined RBAC access policy, can have access to the plaintext data through the decryption operation. 
		\item \emph{Role Management by Multiple organizations}: The roles assigned to users can be managed by more than one organization and can be simultaneously used for data sharing and keyword search operations. 
		\item \emph{Conjunctive Keyword Search}: Users can search for multiple keywords using a single search request.
		\item \emph{Outsourced Decryption}: Users can delegate most of the computationally expensive operation to the public cloud without disclosing any sensitive information.
		
		\item \emph{Prior Authentication}: The public cloud can authenticate a user before performing the costly keyword search and outsourced decryption operations for the user.
		
		\item \emph{Revocation}: Revocation is supported in two following ways:
		\begin{itemize}
			\item \emph{Complete user revocation}: SA can prevent unintended users from accessing its data.
			\item \emph{Role-level user revocation}: SA can revoke one or more roles of a user. The idea is that the revoked user can no longer use the revoked roles for accessing data, while the same user should be able to access data using his/her non-revoked roles if they are qualified enough according to the RBAC access policy.
		\end{itemize}
		
	\end{enumerate}
	\textbf{Security Goals}: The proposed scheme should fulfil the following security requirements:
	\begin{enumerate}
		\item \textit{Data Confidentiality}: Any entity, including the public cloud should not be able to access the plaintext data unless they have proper roles satisfying the defined RBAC access policy. This security notion can be captured by \emph{Semantic Security}. This security notion is also referred to as \emph{Indistinguishability against Chosen Plaintext Attack (IND-CPA).}
		\item \textit{Keyword Secrecy}: Using unqualified search requests or trapdoors, any entity including the public cloud should not be able to learn any useful information about the plaintext keywords associated with the encrypted data. Similarly, any outsider (neither the requesting user nor the public cloud) should be able to learn any useful information about the keywords from the trapdoors. These two security notions can be captured by \emph{Keyword Semantic Security}. This security notion is also referred to as \emph{Indistinguishability against Chosen Keyword Attack (IND-CKA).}
		\item \textit{Forward and Backward Secrecy}: Forward secrecy represents that any new user having qualified roles should be able to decrypt the ciphertexts which are encrypted before he/she joined the system. Backward secrecy represents that a revoked user should not be able to decrypt the ciphertexts which are published after his/her revocation using the revoked roles.
		
		\item \textit{Resistance against Replay Attacks}: If one or more valid trapdoor is exposed to an adversary, the adversary should not be able to launch replay attacks. Many recent keyword search schemes, e.g., \cite{Sun2016, Hu2017} are susceptible to replay attacks if the trapdoors are exposed, as the adversary can re-use the exposed trapdoors using a fresh random number each time she/he wants to perform a keyword search. 
	\end{enumerate}
	\begin{table}[t]
		\tabcolsep 1.0pt
		\centering
		\caption{NOTATIONS}
		\begin{tabular}{p{2cm}p{6.9cm}}
			\hline
			Notation  & Description
			\\[0.5ex]    \hline
			$q$ & a large prime number   \\
			$\mathbb{G}_1, \mathbb{G}_T$ & two cyclic multiplicative groups of order $q$   \\
			$H_1(.), H_2(.)$ & hash functions $H_1: \{0, 1\}^*\rightarrow \mathbb{Z}_q^*$ and $H_2: \mathbb{G}_1\rightarrow \mathbb{Z}_q^*$\\
			$\Phi$ & set of system authorities in the system\\
			$m$ & total number of system authorities in the system\\
			$\Psi_k$ & set of roles associated with a role hierarchy of the $k^{th}$ system authority\\
			$\Gamma$ & set of all roles associated with a ciphertext\\
			$\Gamma_{\Phi}$ & system authorities associated with a ciphertext\\
			$\mathbb{S}_{\mathtt{ID_u}}$& set of roles associated with the user $\mathtt{ID_u}$\\
			$r_{k, i}$ & $i^{th}$ role managed by $k^{th}$ system authority\\
			$\mathbb{R}_{k, i}$ & the set of ancestor roles of $r_{k, i}$ \\
			$\mathtt{ID_u}$ & unique identity of the $u^{th}$ user\\
			$\mathtt{ID_c}$ & unique identity of the public cloud\\
			$\mathtt{RM_{r^k_i}}$ & role-manager which manages role $r^k_i$\\
			$ts$ & current timestamp\\
		\end{tabular}
		\label{notation}
	\end{table}
	\subsection{Framework} 
	\label{franework}
	Broadly the proposed scheme is divided into nine main phases, namely, \emph{System Setup, Management of Roles, Public Cloud Key Generation, New User Enrolment, Role Assignment, Data Encryption, Trapdoor Generation, Data Search}, and \emph{Decryption}. SAs initiate the \emph{System Setup} phase to generate mutually agreed public parameters and master secret through the \textsc{SystemSetup} algorithm. SA performs the \emph{Manage of Role} phase to initialize its role hierarchy and generates role related parameters (both public and secret parameters). It also generates proxy re-encryption keys for the public cloud. It consists of the \textsc{ManageRole} algorithm. SA generates  private and public keys for the public cloud in the \emph{Public Cloud Key Generation} phase using the \textsc{PubCloudkeyGen} algorithm. In the \emph{New User Enrolment} phase, SA mainly issues private and public keys for each registered users through the \textsc{UserPrivKeyGen} algorithm. Role-managers perform \emph{Role Assignment} phase, where they assign roles in the form of role-keys to the users based on their responsibilities and profile in the organization. It consists of the \textsc{UserRoleKeyGen} algorithm. In the \emph{Data Encryption} phase, the owner encrypts data and associated keywords using a RBAC access policy. It consists of the \textsc{Enc} algorithm. To perform keyword search as well as outsourced decryption, the users generate trapdoors in the \emph{Trapdoor Generation} phase using the \textsc{TrapGen} algorithm. The public cloud performs the \emph{Data Search} phase, which consists of \textsc{Authentication}, \textsc{KeySearch}, and \textsc{PartialDec} algorithms. In the \textsc{Authentication}, the public cloud authenticates the requesting user and checks freshness of the keyword search request (i.e., trapdoor) to prevent any replay attacks. In the \textsc{KeySearch}, the public cloud performs keyword search operation on the encrypted data using the received trapdoor. In the \textsc{PartialDec}, the public cloud performs outsourced decryption operations. In this algorithm, the public cloud partially decrypts the ciphertexts which are returned by the \textsc{KeySearch} algorithm. Finally, the user performs \emph{Decryption} phase to decrypt all the partially decrypted ciphertexts received from the public cloud. This phase comprises \textsc{Dec} algorithm. A brief overview of the different algorithms of these phases are explained next. The notations used in this paper are shown in Table \ref{notation}.
	\begin{itemize}
		\item \textsc{SystemSetup} $\left(\left(\mathtt{PP, \{MS_k\}_{\forall k\in \Phi}}\right)\leftarrow 1^\Lambda\right)$: It takes a security parameter $\Lambda$ as input. It outputs public parameter $\mathtt{PP}$ and master secret $\mathtt{MS_k}$ for each SA in the system.
		
		\item \textsc{ManageRole} $\Big(\Big(\mathtt{RP_k}, \{\mathbb{PK}_{r^k_i}\}_{\forall r^k_i \in \Psi_k}, \{\mathtt{RS_{r^k_i}}\}_{\forall r^k_i\in \Psi_k},\\ \left\{\left\{\mathtt{PKey^{r^k_w}_{r^k_i}}\right\}_{\forall r^k_w\in \mathbb{R}_{r^k_i}\setminus \{r^k_i\}}\right\}_{\forall r^k_i\in \Psi_k}\Big)\leftarrow \Big(\mathcal{H}, \mathtt{PP}\Big)\Big)$: It takes a role hierarchy $\mathcal{H}$ and public parameter $\mathtt{PP}$ as input. It outputs role secret parameter $\mathtt{RP_k}$, and for each role $r^k_i\in \Psi_k$, it outputs the role public key $\mathbb{PK}_{r^k_i}$, role secret $\mathtt{RS_{r^k_i}}$ and proxy re-encryption keys $\mathtt{PKey^{r^k_w}_{r^k_i}}$.
		
		\item \textsc{PubCloudKeyGen} $\big(\left(\mathtt{Priv^k_c}, \mathtt{Pub^{1k}_c}, \mathtt{Pub^{2k}_c}\right)\leftarrow \left(\mathtt{PP}, \mathtt{MS_k}, \mathtt{ID_c}\right)\big)$: It takes public parameter $\mathtt{PP}$, master secret $\mathtt{MS_k}$ and identity $\mathtt{ID_c}$ of the public cloud as input. It outputs a private key $\mathtt{Priv^k_c}$ and two public keys $(\mathtt{Pub^{1k}_c}, \mathtt{Pub^{2k}_c})$ for the public cloud. 
		
		\item \textsc{UserPrivKeyGen} $\big(\left(\mathtt{SK^k_{ID_u}}, \mathtt{Pub^k_{ID_u}}, \mathtt{US_{ID_u}}\right)\leftarrow \left(\mathtt{MS_k}, \mathtt{PP}, \mathtt{ID_u}\right)\big)$: It takes master secret $\mathtt{MS_k}$, public parameter $\mathtt{PP}$, and unique identity of a user $\mathtt{ID_u}$ as input. It outputs a secret key $\mathtt{SK^k_{ID_u}}$, a public key $\mathtt{Pub^k_{ID_u}}$ and a user secret $\mathtt{US_{ID_u}}$ for the user $\mathtt{ID_u}$. 
		
		\item \textsc{UserRoleKeyGen} $((\mathtt{RK^{1, u}_{r^k_x}}, \mathtt{RK^{2, u}_{r^k_x}})\leftarrow (\mathtt{PP}, \mathtt{US_{ID_u}}, \mathtt{RS_{r^k_x}}, t_{r^k_x}))$: It takes public parameter $\mathtt{PP}$, user secret $\mathtt{US_{ID_u}}$, and role secret $\mathtt{RS_{r^k_x}}$, role related secret $t_{r^k_x}\in \mathbb{Z}_q^*$ of $r^k_x$ as input. It outputs two role-keys $(\mathtt{RK^{1, u}_{r^k_x}}, \mathtt{RK^{2, u}_{r^k_x}})$ associated with the role $r^k_x$ for the user $\mathtt{ID_u}$.
		
		\item \textsc{Enc} $\Big(\mathbb{CT}\leftarrow \left(\mathtt{PP}, \mathtt{Pub^{1k}_{c}}, \mathtt{Pub^{2k}_c}, \mathtt{M}, \mathbb{W}, \Gamma, \Gamma_{\Phi}\right)\Big)$: It takes public parameter $\mathtt{PP}$, both the public keys $(\mathtt{Pub^{1k}_c}, \mathtt{Pub^{2k}_c})$ of the public cloud, actual plaintext message $\mathtt{M}$, keyword set $\mathbb{W}$ (associated with the actual plaintext message $\mathtt{M}$), a RBAC access policy $\Gamma$, and a set $\Gamma_{\Phi}$ of SAs which are associated with $\Gamma$ as input. It outputs a ciphertext $\mathbb{CT}$.
		
		\item \textsc{TrapGen} $(\left(\mathtt{Trap}, v\right)\leftarrow (\{\mathtt{RK^{1, u}_{r^k_x}}, \mathtt{RK^{2, u}_{r^k_x}}\}_{\forall r^k_x\in \mathbb{S}_{\mathtt{ID_u}}}, \mathtt{SK^k_{ID_u}}, \mathbb{S}_{\mathtt{ID_u}}, w))$: It takes both the role-keys $(\mathtt{RK^{1, u}_{r^k_x}}, \mathtt{RK^{2, u}_{r^k_x}})$, secret key $\mathtt{SK^k_{ID_u}}$, user role set $\mathbb{S}_{\mathtt{ID_u}}$  of a user $\mathtt{ID_u}$, and keyword $w$ as input. It outputs a trapdoor $\mathtt{Trap}$ and a random number $v\in \mathbb{Z}_q^*$.
		
		\item \textsc{Authentication} $\Big(\left(V^1_3/\perp\right)\leftarrow \Big(\{\mathtt{Priv^k_c}\}_{\mathtt{\forall k\in \Gamma_{\Phi}}}, \mathtt{Trap}, \mathtt{Pub^k_{ID_u}}, \mathtt{ID_u}, ts'\Big)\Big)$: It takes private keys of the public cloud $\mathtt{Priv^k_c}$ issues by all the system authorities in the set $\Gamma_{\Phi}$, trapdoor $\mathtt{Trap}$, public key $\mathtt{Pub^k_{ID_u}}$ of a user $\mathtt{ID_u}$, identity $\mathtt{ID_u}$ of the user, and current timestamp $ts'$ as input. If the user $\mathtt{ID_u}$ is legitimate and the trapdoor was not previously issued, it outputs $V_3^1$ for a successful authentication. Otherwise, it outputs $\perp$ which represents either an unsuccessful authentication or an invalid trapdoor. 
		
		\item \textsc{KeySearch} $((\mathbb{CT}/\perp)\leftarrow (\mathbb{CT}, \mathtt{Trap}, V^3_1))$: It takes trapdoor $\mathtt{Trap}$, $V^3_1$ and a ciphertext $\mathbb{CT}$ as input. It outputs the ciphertext $\mathbb{CT}$ if and only if for all $r^k_i\in \Gamma$ there is $r^k_x \in \mathbb{S}_{\mathtt{ID_u}}$ such that $r^k_x\in \mathbb{R}_{r^k_i}$ and the keyword $w$ associated with the trapdoor has a match with a keyword associated with the ciphertext $\mathbb{CT}$. Otherwise, it outputs $\perp$, which represents an unsuccessful search operation.
		
		\item \textsc{PartialDec} $\left(\mathbb{CT}'\leftarrow \left(\mathbb{CT}, \mathtt{Trap}, \{\mathtt{Priv^k_c}\}_{\forall k\in \Gamma_{\Phi}}, \mathbb{S}_{\mathtt{ID_u}}\right)\right)$: It takes the ciphertext $\mathbb{CT}$, trapdoor $\mathtt{Trap}$, private keys $\mathtt{Priv^k_c}$ of the public cloud associated with the system authorities in $\Gamma_{\Phi}$, and user role set $\mathbb{S}_{\mathtt{ID_u}}$ as input. It outputs a partially decrypted ciphertext $\mathbb{CT}'$.
		
		\item \textsc{FullDEC}$(\mathtt{M}\leftarrow (\mathbb{CT}', \mathtt{Priv_{ID_u}}, v))$: It takes the partially decrypted ciphertext $\mathbb{CT}'$, user private key $\mathtt{Priv_{ID_u}}$, and $v$ as input and outputs the actual plaintext message $\mathtt{M}$. 
	\end{itemize}

	\subsection{Security Model}
	\label{security_model}
	The two games, namely, \emph{Semantic Security against Chosen Plaintext Attack} (IND-CPA) and \emph{Semantic Security against Chosen Keyword Attack} (IND-CKA) are used to define the security model of the proposed scheme. These two games are defined next.
	\subsubsection{Semantic Security against Chosen Plaintext Attack}
	\label{CPA}
	The semantic security of the proposed scheme defined on \emph{Chosen Plaintext Attack} (CPA) security under \emph{Selective-ID Model}\footnote{In the Selective-ID security model, the adversary must submit a set of challenged roles before starting the security game. This is essential in our security proof to set up the role public key (please refer Section \ref{security_analysis} for more details).}. The CPA security can be illustrated using the following security game IND-CPA between a challenger $\mathcal{C}$ and an adversary $\mathcal{A}_1$.
	\par 
	\textsc{\textbf{Init}} Adversary $\mathcal{A}_1$ sends a challenged role set $\Gamma^*$, a keyword $w$ and two identities $\mathtt{ID_u^*}, \mathtt{ID^*_c}$ to the challenger $\mathcal{C}$. 
	
	\textsc{\textbf{Setup}} Challenger runs the \textsc{SystemSetup} algorithm to generate public parameters and master secrets. Challenger $\mathcal{C}$ generates role public keys, role secrets and proxy re-encryption keys using the \textsc{ManageRole} algorithm. It also generates public and private keys using the \textsc{PubCloudKeyGen} and \textsc{UserPrivKeyGen} algorithms. Challenger $\mathcal{C}$ sends the public parameter, role public keys, proxy re-encryption keys, public and private keys to the adversary $\mathcal{A}_1$. It keeps the master secret and role secrets in a secure place. 
	
	\par
	\textsc{\textbf{Phase 1}} Adversary $\mathcal{A}_1$ submits a role set $\mathbb{S}^*$ to the challenger $\mathcal{C}$ for role-keys so that there exits at least one role $r^k_x\in \mathbb{S}^*$ such that $r^k_x\notin \mathbb{R}_{r^k_i}$, where $r^k_i\in \Gamma^*$. Challenger $\mathcal{C}$ runs the \textsc{userRoleKeyGen} algorithm to generate role-keys for the adversary $\mathcal{A}_1$. Adversary $\mathcal{A}_1$ can send queries for the role-keys to the challenger $\mathcal{C}$ by polynomially many times.


	\par 
	\textsc{\textbf{Challenge}} When adversary $\mathcal{A}_1$ decides that \textsc{Phase 1} is over, it submits two equal length messages $\mathtt{K_0}$ and $\mathtt{K_1}$, which were not challenged before, to the challenger $\mathcal{C}$. Challenger $\mathcal{C}$ flips a random binary coin $\omega$ and encrypts message $\mathtt{K_\omega}$ using the \textsc{Enc} algorithm for the challenged role set $\Gamma^*$. Challenger $\mathcal{C}$ sends the encrypted message of $\mathtt{K_\omega}$ to adversary $\mathcal{A}_1$.

	\par 
	\textsc{\textbf{Phase 2}} Same as \textsc{\textbf{Phase 1}}.
	
	\par 
	\textsc{\textbf{Guess}} Adversary $\mathcal{A}_1$ outputs a guess $\omega'$ of $\omega$. The advantage of winning this game for adversary $\mathcal{A}_1$ is $Adv^{IND-CPA}_{\mathcal{A}_1}= \left|Pr[\omega'= \omega]- \frac{1}{2}\right|$. 
	
	\begin{definition}
		The proposed scheme is secure against chosen plaintext attack if $Adv^{IND-CPA}_{\mathcal{A}_1}$ is negligible for any polynomial time adversary $\mathcal{A}_1$.
	\end{definition}

	\subsubsection{Semantic Security against Chosen Keyword Attack}
	\label{CKA}
	The semantic security of the proposed keyword search scheme defined on Chosen Keyword Attack (CKA) security under the same \emph{Selective ID Model} as described in Section \ref{CPA}. The CKA security can be demonstrated using the following security game IND-CKA between a challenger $\mathcal{C}$ and an adversary $\mathcal{A}_2$.
	
	\par 
	\textsc{\textbf{Init}} Adversary $\mathcal{A}_2$ sends a set of challenged roles $\Gamma^*$ and two identities $\mathtt{ID_u^*}, \mathtt{ID^*_c}$ to the challenger $\mathcal{C}$. 
	
	\textsc{\textbf{Setup}} Challenger runs the \textsc{SystemSetup} algorithm to generate public parameters and master secrets. Challenger $\mathcal{C}$ generates role public keys, role secrets and proxy re-encryption keys using the \textsc{ManageRole} algorithm. It also generates public and private keys using \textsc{PubCloudKeyGen} algorithm and a public key using \textsc{UserPrivKeyGen} algorithm. Challenger $\mathcal{C}$ sends the public parameter, role public keys, proxy re-encryption keys, public and private keys to the adversary $\mathcal{A}_2$. It keeps the master secret and role secrets in a secure place. 
	
	\par
	\textsc{\textbf{Phase 1}} Adversary $\mathcal{A}_2$ submits a set of roles $\mathbb{S}^*$ and a keyword $w$ to the challenger $\mathcal{C}$ so that there exits at least one role $r^k_x\in \mathbb{S}^*$ such that $r^k_x\notin \mathbb{R}_{r^k_i}$, where $r^k_i\in \Gamma^*$. Challenger initiates the \textsc{TrapGen} algorithm to generate a trapdoor for the adversary $\mathcal{A}_2$. Finally, challenger $\mathcal{C}$ sends the generated trapdoor to the adversary $\mathcal{A}_2$. Afterwards, adversary $\mathcal{A}_2$ can send queries for the trapdoor to the challenger $\mathcal{C}$ by polynomially many times.
	
	\par 
	\textsc{\textbf{Challenge}} When adversary $\mathcal{A}_2$ decides that \textsc{Phase 1} is completed, it submits two equal length keywords $w_0$ and $w_1$, which were not challenged before, to the challenger $\mathcal{C}$. Challenger $\mathcal{C}$ flips a binary coin $\omega$ and encrypts keyword $w_\omega$ using the \textsc{Enc} algorithm for the challenged role set $\Gamma^*$. Challenger $\mathcal{C}$ sends the encrypted ciphertext of $w_\omega$ to the adversary $\mathcal{A}_2$.

	\par 
	\textsc{\textbf{Phase 2}} Same as \textsc{\textbf{Phase 1}}.
	
	\par 
	\textsc{\textbf{Guess}} Adversary $\mathcal{A}_2$ outputs a guess $\omega'$ of $\omega$. The advantage of winning this game for adversary $\mathcal{A}_2$ is $Adv^{IND-CKA}_{\mathcal{A}_2}= \left|Pr[\omega'= \omega]- \frac{1}{2}\right|$. 
	
	\begin{definition}
		The proposed scheme is secure against the chosen keyword attack if $Adv^{IND-CKA}_{\mathcal{A}_2}$ is negligible for any polynomial time adversary $\mathcal{A}_2$.
	\end{definition}
	
	\section{Preliminaries}
	\label{preli}
	This section presents an overview of a role hierarchy and bilinear pairing. It also presents an overview of a group key distribution mechanism and a mathematical assumption which is used in this paper.
	\begin{figure}[t]
		\centering
		\begin{subfigure}[!t]{0.16\textwidth}
			\includegraphics[width=\textwidth]{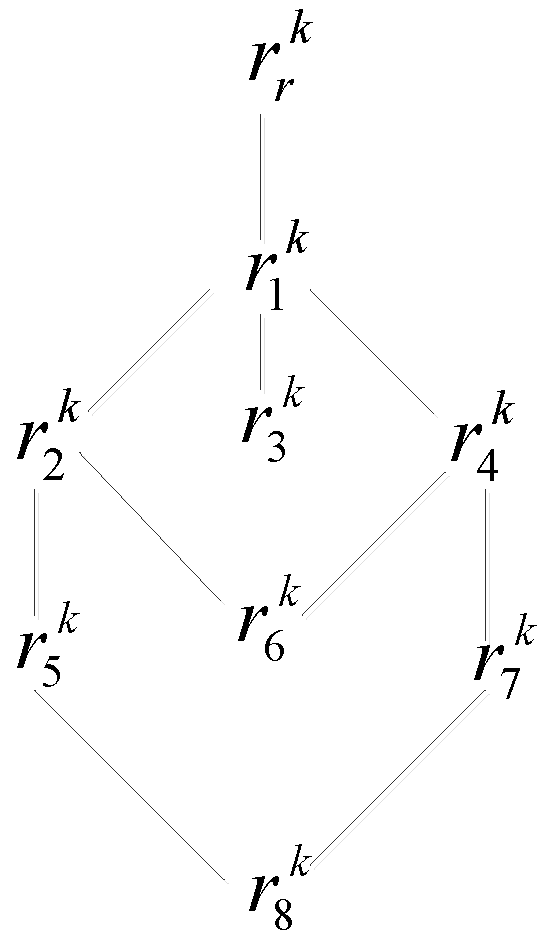}
			\caption{Role Hierarchy 1}
			\label{fig:RH1}
		\end{subfigure}
		~ 
		\begin{subfigure}[!t]{0.23\textwidth}
			\includegraphics[width=\textwidth]{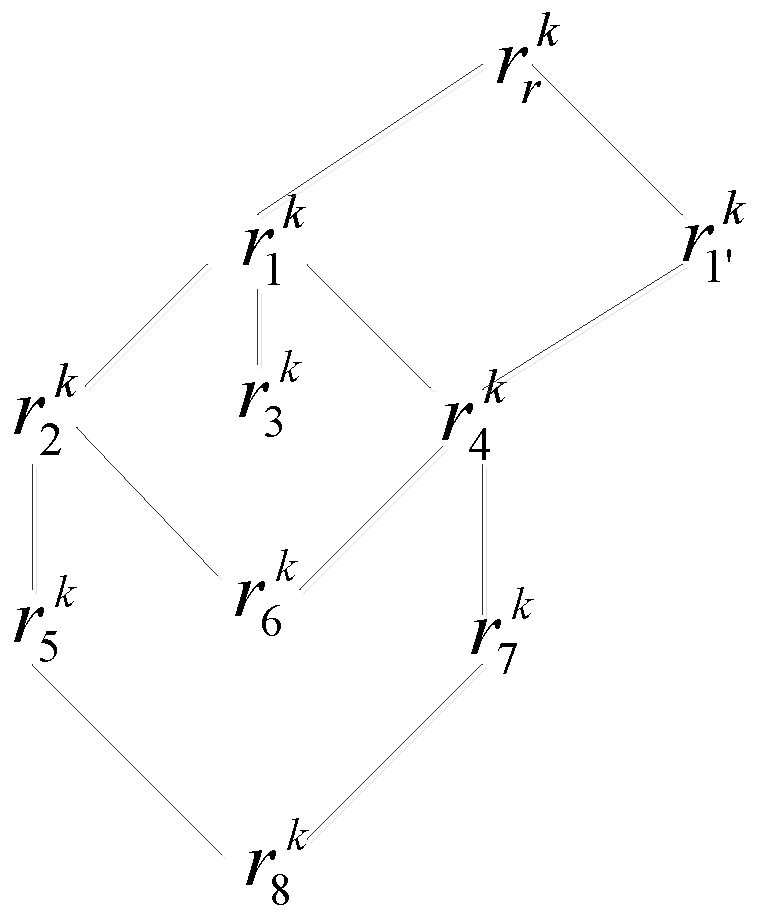}
			\caption{Role Hierarchy 2}
			\label{fig:RH2}           
		\end{subfigure}
		\caption{Sample Role Hierarchy (RH)}
		\label{fig:RH}
	\end{figure}

	\subsection{Role Hierarchy notations}
	\label{role_hierarchy}
	
	In the proposed scheme, roles are organized in a hierarchy where ancestor roles can inherit access privileges of its descendant roles. Figure \ref{fig:RH} shows two sample role hierarchies, namely Role Hierarchy 1 (Figure \ref{fig:RH1}) and Role Hierarchy 2 (Figure \ref{fig:RH2}). We consider Role Hierarchy 1 (Figure \ref{fig:RH1}) as an example to define the following notations of a role hierarchy.
	\begin{itemize}
		\item $r^k_r$: root role of a role hierarchy. We assume that in any role hierarchy there can be only one root role. 
		\item $\Psi_k$: set of all roles in the role hierarchy. For example, $\Psi_k= \{r^k_r, r^k_1, r^k_2, r^k_3, r^k_4, r^k_5, r^k_6, r^k_7, r^k_8\}$
		\item $\mathbb{R}_{r^k_i}$: ancestor set of the role $r^k_i$. For example, $\mathbb{R}_{r^k_8}= \{r^k_r, r^k_1, r^k_2, r^k_4, r^k_5, r^k_6, r^k_7, r^k_8\}, \mathbb{R}_{r^k_5}=\{r^k_r, r^k_1, r^k_2, r^k_5\}$ and $\mathbb{R}_{r^k_6}=\{r^k_r, r^k_1, r^k_2, r^k_4, r^k_6\}$.
	\end{itemize}
	
	\subsection{Bilinear Pairing}
	\label{pairing}
	Let $\mathbb{G}_1$ and $\mathbb{G}_T$ be two cyclic multiplicative groups of order $q$. Let $g$ be a generator of $\mathbb{G}_1$. The bilinear map $\hat{e}: \mathbb{G}_1\times \mathbb{G}_1\rightarrow \mathbb{G}_T$ has the following properties:
	\begin{itemize}
		\item \textit{Bilinear}: $\hat{e}\left(g^a, g^b\right)= \hat{e}\left(g, g\right)^{ab}$, $\forall g\in \mathbb{G}_1$ and $\forall(a, b)\in \mathbb{Z}_q^*$
		\item \textit{Non-degenerate}: $\hat{e}\left(g, g\right)\ne 1$
		\item \textit{Computable}: $\hat{e}(g, g)$ is efficiently computable for all $g\in \mathbb{G}_1$
	\end{itemize}
	\begin{figure*}[t]
		\centering
		\scalebox{5}{\includegraphics[width=2.3cm, height=1.8cm]{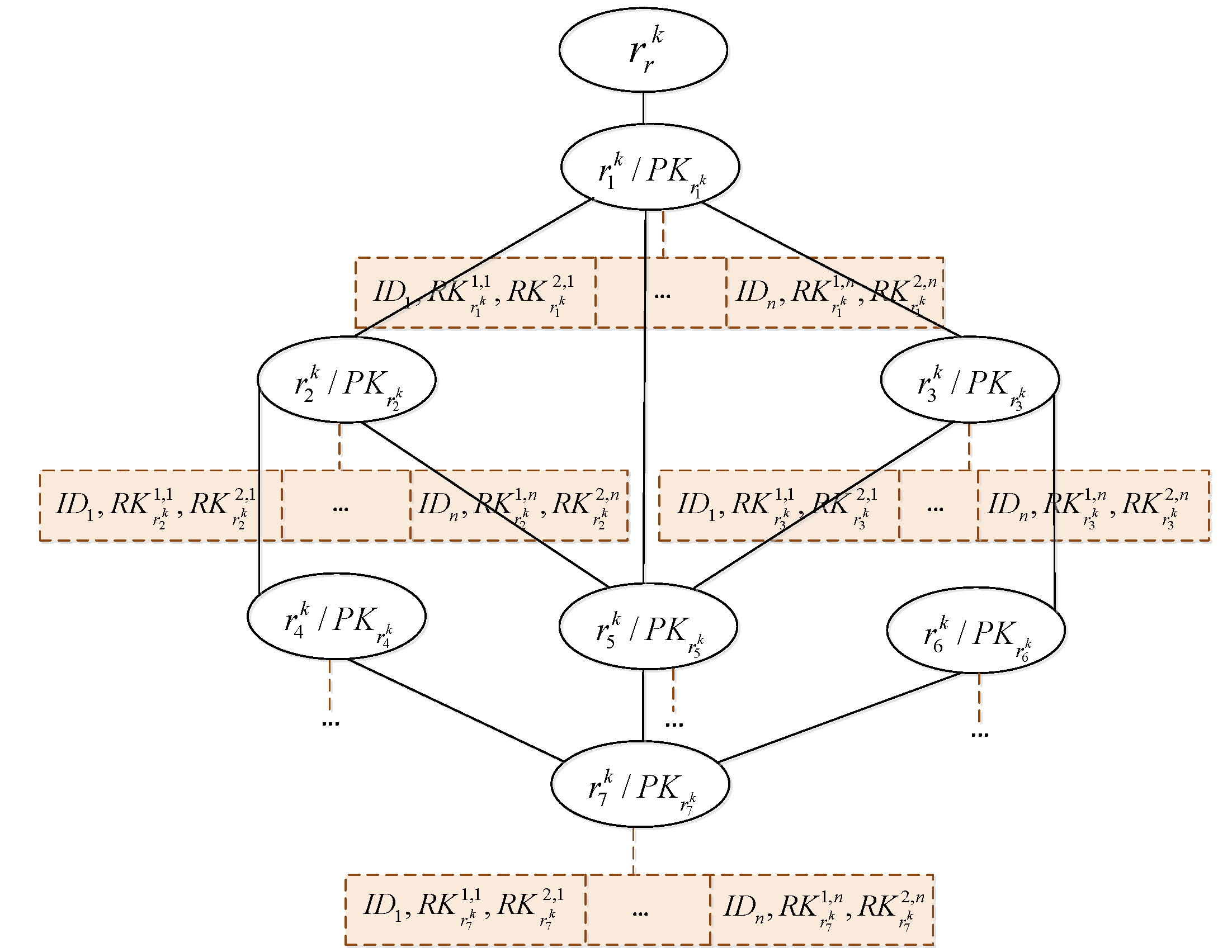}}
		\caption{Sample Role Key Hierarchy (RKH)}
		\label{RKH}
	\end{figure*}
	\subsection{Group Key Distribution}
	\label{key_distribution}
	In \cite{Burmester2005}, Burmester \emph{et al.} proposed a two round group key distribution scheme using the concept of Diffie-Hellman assumption. Their scheme works as follows:
	\par
	Let $\mathbb{U}= \{\mathtt{ID_1}, \mathtt{ID_2}, ..., \mathtt{ID_n}\}$ be the group of $n$ users. Suppose the users are arranged into a cycle. To compute a group key among the users, each user $\mathtt{ID_i}\in \mathbb{U}$ selects a random secret number $a_i\in \mathbb{Z}_q^*$ and broadcasts $x_i= g^{a_i}$ where $g$ is a generator of group $\mathbb{G}_1$. Afterward, it publishes $X_i= (\frac{x_{i+1}}{x_{i-1}})^{a_i}$. Finally, each user $\mathtt{ID_i}$ in the group computes a common key $\mathtt{CK}= g^{a_1\cdot a_2+ a_2\cdot a_3+ ...+ a_n\cdot a_1}$ without knowing others' secrets and without disclosing the common key to any other unintended entities.
	
	\subsection{Decisional Bilinear Diffie-Hellman (DBDH)}
	\label{complexity_assumption}
	Let $\mathbb{G}_1$ and $\mathbb{G}_T$ be two cyclic multiplicative groups of order $q$. Let $g$ be a generator of $\mathbb{G}_1$ and $\hat{e}: \mathbb{G}_1\times \mathbb{G}_1\rightarrow \mathbb{G}_T$ be an efficiently computable non-degenerate bilinear map. 
	The Decisional Bilinear Diffie-Hellman (DBDH) Assumption is defined as follows: 
	\label{DBDH}
	No probabilistic polynomial time adversary is able to distinguish the tuples $\left<g, g^a, g^b, g^c, Z= \hat{e}\left(g, g\right)^{abc}\right>$ and $\left<g, g^a, g^b, g^c, Z=\hat{e}\left(g, g\right)^{z}\right>$ with non-negligible advantage, where $(a, b, c, z)\in \mathbb{Z}_q^*$ are randomly chosen.  
	
	\section{Proposed Scheme}
	\label{proposed_scheme}
	This section presents the proposed scheme in details. First, a brief overview of the proposed scheme is presented, followed by its main construction.
	
	\subsection{Overview}
	The main goal of the proposed scheme is to enable the owners to enforce RBAC access policies on the encrypted data so that only the users with the authorized roles can perform the keyword search along with efficient data decryption. To achieve this, the proposed scheme devises a novel RBE technique that enables only the users having authorized roles satisfying the specified RBAC access policy to delegate the keyword search capability to the public cloud without disclosing any sensitive information. To reduce decryption cost at the user side, the devised RBE technique also enables the authorized users to delegate computationally expensive cryptographic operations to the public cloud. 
	\par 
	In the proposed scheme, each organization is allowed to maintain its own role hierarchy, and each role hierarchy is associated with a Role-Key Hierarchy (RKH). In Figure \ref{RKH}, a sample RKH is shown. Each node in a RKH represents a role, and each role, say $r^k_i$, is associated with a role public key, say $\mathtt{PK_{r^k_i}}$. In addition, each role is associated with a set of users who hav
	
	In the proposed scheme, each organization is allowed to maintain its own role hierarchy, and each role hierarchy is associated with a Role-Key Hierarchy (RKH). In Figure \ref{RKH}, a sample RKH is shown. Each node in a RKH represents a role, and each role (except the root role), say $r^k_i$, is associated with a role public key, say $\mathtt{PK_{r^k_i}}$. In addition, each role (except the root role) is associated with a set of users who have that role, and the users are assigned with a unique pair of role-keys for each role they possess\footnote{In our proposed scheme, the root role ($r^k_r$) is not assigned to any users and is internally managed by the SA. As such, we do not consider any user set with the root role in Figure \ref{RKH}. More details are given in the following sections.}. The role-keys are generated in such a way that the user can use them to compute trapdoors to perform a keyword search over the ciphertexts, which are encrypted using a role public key of any descendent role. The same trapdoor can also be used to perform the outsourced decryption operation. This in turn enables the users to gain access to the actual plaintext data. This process is illustrated as follows. Let us assume that the owner wants to authorize all the users having access privileges for the role $r^k_5$ to have access to data. The owner encrypts the data and the associated keywords using the role public key $\mathtt{PK_{r^k_5}}$. Any user who possesses any one of the roles in $\mathbb{R}_{r^k_5}=\{r^k_r, r^k_1, r^k_2, r^k_3, r^k_5\}$ can search and decrypt the encrypted data using their respective role-keys. That is, the user possesses a qualified role for accessing the ciphertext. Similarly, if the owner encrypts data and associated keywords using the role public keys $\mathtt{PK_{r^k_1}}$ and $\mathtt{PK_{r^k_4}}$, then any user who possesses roles in $\mathbb{R}_{r^k_1}=\{r^k_r, r^k_1\}$ and $\mathbb{R}_{r^k_4}=\{r^k_r, r^k_1, r^k_2, r^k_4\}$ respectively can perform keyword search and data decryption using their respective role-keys.

	\par 
	To support multi-organization data sharing, the proposed scheme takes advantage of an existing group key distribution protocol to generate a common master secret for all the participating organizations. This master secret is used for generating the system parameters, including public parameters and master secrets of each organizations. This allows a user to possess more than one role from different organizations. More details are given in the following subsection.

	\subsection{Construction}
	\label{construction}
	A detailed description of all the phases of the proposed scheme is presented as follows.

	\subsubsection{System Setup}
	\label{system_setup}
	In this phase, the system authority of each organization mutually publishes the system public parameter, and they generate their own master secrets. This phase consists of the \textsc{SystemSetup} algorithm which is defined next.
	
	\paragraph{\textsc{SystemSetup} $\left(\left(\mathtt{PP, \{MS_k\}_{\forall k\in \Phi}}\right)\leftarrow 1^\Lambda\right)$} It chooses two cyclic multiplicative bilinear groups $\mathbb{G}_1$ and $\mathbb{G}_T$ of order $q$, where $q$ is a large prime number. It also chooses a generator $g\in \mathbb{G}_1$, random numbers $\{\eta_k, \mu_k, \mathtt{x_k}\}_{\forall k\in \Phi}\in \mathbb{Z}_q^*$ and two hash functions $H_1: \{0, 1\}^*\rightarrow \mathbb{Z}_q^*, H_2: \mathbb{G}_1\rightarrow \mathbb{Z}_q^*$. Afterward, all the system authorities follow a group key generation protocol, as described in Section \ref{key_distribution}, to compute a shared secret $g^{y}$, where $y=  y_1\cdot y_2+ y_2\cdot y_3+ ...+ y_{m}\cdot y_1$ and $m$ is the total number of system authorities. Afterward, it computes $Y= \hat{e}(g, g)^y$ and $h^k_1= g^{\eta_k}$, and then publishes the system public parameter $\mathtt{PP}= \left<\mathbb{G}_1, \mathbb{G}_T, g, \hat{e}, H_1, H_2, Y, \{h^k_1\}_{\forall k\in \Phi}\right>$. Each system authority, say $k^{th}$ system authority $\mathtt{SA_k}$, keeps master secret $\mathtt{MS_k}=\left<g^y,\eta_k, \mu_k, \mathtt{x_k}\right>$ in a secure place.
	\begin{remark}
		All the system authorities can check validity of $Y$ by comparing $\hat{e}(g^y, g)\stackrel{?}{=}Y$. Also, any number of new system authorities can be added in the system at any time by sharing the existing group secret key, i.e., $g^y$.
	\end{remark}

	\subsubsection{Management of Roles}
	\label{manage_role}
	In this phase, a system authority generates the role related parameters. Suppose the system authority $\mathtt{SA_k}$ wants to initialize a role hierarchy $\mathcal{H}$. The system authority $\mathtt{SA_k}$ generates role secrets $\mathtt{RS_{r^k_i}}$ and role public keys $\mathbb{PK}_{r^k_i}$ for each role $r^k_i$ associated with $\mathcal{H}$. It also computes proxy re-encryption keys $\{\mathtt{PKey^{r^k_x}_{r^k_i}}\}_{r^k_x\in \mathbb{R}_{r^k_i}\setminus\{r^k_i\}}$ for each role $r^k_i$ (except the root role) associated with the role hierarchy $\mathcal{H}$. It stores the role public keys in its public bulletin board and keeps the role secrets in a secure place. It also shares each role secret to its corresponding role-manager. That is, the role secret associated with $r^k_i$, i.e., $\mathtt{RS_{r^k_i}}$ is shared with the role-manager which manages $r^k_i$, i.e., $\mathtt{RM_{r^k_i}}$. Moreover, the proxy re-encryption keys are sent to the proxy-server (i.e., public cloud) using secure-channels. This phase consists of the \textsc{ManageRole} algorithm which is defined next.
	
	\paragraph{\textsc{ManageRole} $\Big(\Big(\mathtt{RP_k}, \{\mathbb{PK}_{r^k_i}\}_{\forall r^k_i \in \Psi_k}, \{\mathtt{RS_{r^k_i}}\}_{\forall r^k_i\in \Psi_k},\\ \left\{\left\{\mathtt{PKey^{r^k_w}_{r^k_i}}\right\}_{\forall r^k_w\in \mathbb{R}_{r^k_i}\setminus \{r^k_i\}}\right\}_{\forall r^k_i\in \Psi_k}\Big)\leftarrow \Big(\mathcal{H}, \mathtt{PP}\Big)\Big)$} It selects random numbers $\{t_{r^k_i}\}_{\forall r^k_i\in \Psi_k}\in \mathbb{Z}_q^*$. It computes role secrets $\mathtt{RS_{r^k_i}}$, role public key $\mathbb{PK}_{r^k_i}= \left<\mathtt{PK_{r^k_i}}, r^k_i, \mathbb{R}_{r^k_i}\right>$ and proxy re-encryption key $\{\mathtt{PKey^{r^k_w}_{r^k_i}}\}_{\forall r^k_w\in \mathbb{R}_{r^k_i}\setminus \{r^k_i\}}$ for each role $r^k_i\in (\Psi_k\setminus\{r^k_r\})$, where
	\begin{align}
	\mathtt{RS_{r^k_i}}= & \prod_{\forall r^k_j\in \mathbb{R}_{r^k_i}}t_{r^k_j}\\
	\mathtt{PK_{r^k_i}}= & g^{\prod_{\forall r^k_j\in \mathbb{R}_{r^k_i}}t_{r^k_j}} \\
	\mathtt{PKey^{r^k_w}_{r^k_i}}= & \prod_{\forall r^k_j\in \mathbb{R}_{r^k_i}\setminus \{r^k_w\}} t_{r^k_j}\\
	\end{align}
	$\mathbb{R}_{r^k_i}$ is the set of ancestor roles of $r^k_i$ and role secret parameter $\mathtt{RP_k}= \left<\{t_{r^k_i}\}_{\forall r^k_i\in\Psi_k}\right>$. The system authority sends each secret role parameter and role secret associated with a role to the role-manager which is responsible of its management. For example, secret role parameter $t_{r^k_i}$ and role secret $\mathtt{RS_{r^k_i}}$ are shared with the role-manager $\mathtt{RM_{r^k_i}}$.  Note that the root role is internally managed by the system authority. As such, no proxy re-encryption key, role secret key, role public key are generated for the root role. 
	
	\subsubsection{Public Cloud Key Generation}
	\label{clou_key_gen}
	In this phase, a system authority generates keys for the public cloud. Let the system authority $\mathtt{SA_k}$ wants to issue keys for the public cloud. It computes a private key $\mathtt{Priv^k_c}$, two public keys $(\mathtt{Pub^{1k}_c}, \mathtt{Pub^{2k}_c})$ and sends the private key $\mathtt{Priv^k_c}$ to the public cloud using a secure-channel. It stores both the public keys $(\mathtt{Pub^{1k}_c}, \mathtt{Pub^{2k}_c})$ in its public bulletin board. This phase consists of the \textsc{PubCloudKeyGen} algorithm which is defined next.
	
	\paragraph{\textsc{PubCloudKeyGen} $\big(\left(\mathtt{Priv^k_c}, \mathtt{Pub^{1k}_c}, \mathtt{Pub^{2k}_c}\right)\leftarrow \left(\mathtt{PP}, \mathtt{MS_k}, \mathtt{ID_c}\right)\big)$} It computes a private key $\mathtt{Priv^k_c}$ and two public keys $(\mathtt{Pub^{1k}_c}, \mathtt{Pub^{2k}_c})$ for the public cloud as follows:
	\begin{align}
	\mathtt{Priv^k_c}= & H_2\big(\left(g^y\right)^{\frac{H_1(\mathtt{ID_c})}{\mathtt{x_k}}}\big)= H_2\big(g^{\frac{y\cdot H_1(\mathtt{ID_c})}{\mathtt{x_k}}}\big)\\
	\mathtt{Pub^{1k}_c}= & g^{\mu_k\cdot \mathtt{Priv^k_c}}\\
	\mathtt{Pub^{2k}_c}= & g^{\mathtt{x_k}\cdot \mathtt{Priv^k_c}}
	\end{align}
	\subsubsection{New User Enrolment}
	\label{user_enrolment}
	A system authority initiates this phase when a new legitimate user, say $\mathtt{ID_u}$, wants to join an organization, say $k^{th}$ organization. The system authority $\mathtt{SA_k}$ generates a secret key $\mathtt{SK^k_{ID_u}}$ and public key $\mathtt{Pub^k_{ID_u}}$ for the user $\mathtt{ID_u}$. It also generates a user secret $\mathtt{US_{ID_u}}$ which is shared with all the role-managers under its control. $\mathtt{SA_k}$ sends the secret key $\mathtt{SK^k_{ID_u}}$ to the user $\mathtt{ID_u}$ using a secure-channel and keeps the public key $\mathtt{Pub^k_{ID_u}}$ in its public bulletin board. This phase comprises the \textsc{UserPrivKeyGen} algorithm which is defined next.
	\paragraph{\textsc{UserPrivKeyGen} $\big(\left(\mathtt{SK^k_{ID_u}}, \mathtt{Pub^k_{ID_u}}, \mathtt{US_{ID_u}}\right)\leftarrow \left(\mathtt{MS_k}, \mathtt{PP}, \mathtt{ID_u}\right)\big)$} It issues a pair of secret key $\mathtt{SK^k_{ID_u}}= \left<\mathtt{Priv_{ID_u}}, \mathtt{Priv^{k}_{ID_u}}\right>$, public key $\mathtt{Pub^k_{ID_u}}$, and a user secret $\mathtt{US_{ID_u}}$ as follows: 
	\begin{align}
	\mathtt{Priv_{ID_u}}= & H_2\left((g^y)^{H_1(\mathtt{ID_u})}\right)= H_2\left(g^{y\cdot H_1(\mathtt{ID_u})}\right)\\
	\mathtt{Priv^{k}_{ID_u}}= & \left(g^y\right)^{\frac{\mathtt{Priv_{ID_u}}}{\eta_k}}\cdot g^{\frac{\mathtt{x_k}}{\eta_k}}= g^{\frac{y\cdot \mathtt{Priv_{ID_u}}+ \mathtt{x_k}}{\eta_k}}\\
	\mathtt{Pub^{k}_{ID_u}}= & g^{\frac{H_2\left(\mathtt{Priv^{k}_{ID_u}}\right)}{\mathtt{Priv_{ID_u}}}}\\
	\mathtt{US_{ID_u}}= & (g^{y})^{\mathtt{Priv_{ID_u}}}\cdot g^{\mu_k}= g^{y\cdot \mathtt{Priv_{ID_u}}+ \mu_k}
	\end{align}
	Note that all the system authorities compute the same private key $\mathtt{Priv_{ID_u}}$ for the user $\mathtt{ID_u}$. Hence, the user $\mathtt{ID_u}$ needs to keep only one copy of it.
	
	\subsubsection{Role Assignment}
	\label{role_assignment}
	In this phase, a role-manager assigns roles to a legitimate user. Suppose the role-manager $\mathtt{RM_{r^k_x}}$ wants to assign a role $r^k_x$ to the user $\mathtt{ID_u}$. To do so, $\mathtt{RM_{r^k_x}}$ computes two role-keys $(\mathtt{RK^{1, u}_{r^k_x}}, \mathtt{RK^{2, u}_{r^k_x}})$ for the user $\mathtt{ID_u}$ and sends the role-keys to the user $\mathtt{ID_u}$ using a secure-channel. This phase comprises the \textsc{UserRoleKeyGen} algorithm which is described next.
	\paragraph{\textsc{UserRoleKeyGen} $((\mathtt{RK^{1, u}_{r^k_x}}, \mathtt{RK^{2, u}_{r^k_x}})\leftarrow (\mathtt{PP}, \mathtt{US_{ID_u}}, \mathtt{RS_{r^k_x}}, t_{r^k_x}))$} Let's say, user $\mathtt{ID_u}$ is assigned with the role $r^k_x$. $\mathtt{RM_{r^k_x}}$ computes the role-keys $\mathtt{RK^{1, u}_{r^k_x}}$ and $\mathtt{RK^{2, u}_{r^k_x}}$ as follows:
	
	\begin{align}
	\mathtt{RK^{1, u}_{r^k_x}}= &  \left(\mathtt{US_{ID_u}}\right)^{\frac{1}{\mathtt{RS_{r^k_x}}}}= g^{\frac{y\cdot \mathtt{Priv_{ID_u}}+ \mu_k}{\prod_{r^k_j\in \mathbb{R}_{r^k_x}}t_{r^k_j}}}\\
	\mathtt{RK^{2, u}_{r^k_x}}= & \left(\mathtt{US_{ID_u}}\right)^{\frac{1}{t_{r^k_x}}}= g^{\frac{y\cdot \mathtt{Priv_{ID_u}}+ \mu_k}{t_{r^k_x}}}
	\end{align}
	
	\subsubsection{Data Encryption}
	\label{encryption}
	In this phase, the owner encrypts the plaintext data and then outsources the encrypted data to the cloud storage servers. The owner first encrypts the plaintext data using a random symmetric key by following a secure symmetric key encryption algorithm (e.g., Advanced Encryption Standards). The owner then chooses a set of keywords associated with the actual plaintext data and encrypts the chosen keywords along with the symmetric key using our proposed \textsc{ENC} algorithm. Finally, the owner combines both ciphertexts (i.e., symmetric key and actual plaintext data components) into one archive and outsources the archive file to the public cloud. The \textsc{ENC} algorithm is defined as follows:
	
	\paragraph{\textsc{Enc} $\Big(\mathbb{CT}\leftarrow \left(\mathtt{PP}, \mathtt{Pub^{1k}_{c}}, \mathtt{Pub^{2k}_c}, \mathtt{M}, \mathbb{W}, \Gamma, \Gamma_{\Phi}\right)\Big)$} Let an owner of the $k^{th}$ organization wants to share a plaintext message $\mathtt{M}$ with the users who possess access rights for the roles in $\Gamma$. Let $w$ be a keyword from the keyword space $\mathbb{W}$. First, the owner chooses a random number $\mathtt{K}\in \mathbb{G}_T$ and encrypts the plaintext message $\mathtt{M}$ using $\mathtt{K}$ by following a symmetric key encryption algorithm. Afterward, the owner encrypts the random number $\mathtt{K}$ along with the keyword $w$ using the role public parameters of the roles in $\Gamma$. 
	\par 
	The owner chooses random numbers $(\{d_{r^k_i}, d' _{r^k_i}\}_{\forall r^k_i\in \Gamma})\in \mathbb{Z}_q^*$, where $d_i= \sum_{r^k_i\in \Gamma} d_{r^k_i}$, $d_j= \sum_{r^k_i\in \Gamma} d'_{r^k_i}$ and $d= d_i+ d_j$. The owner also computes $\{d_k= \sum_{\forall r^k_i\in (\Gamma \cap \Psi_k)}d_{r^k_i}\}_{\forall k\in \Gamma_{\Phi}}$ and $\{d'_k= \sum_{\forall r^k_i\in (\Gamma \cap \Psi_k)}d'_{r^k_i}\}_{\forall k\in \Gamma_{\Phi}}$. Finally, the owner generates a ciphertext $\mathbb{CT}= \left<\mathtt{Enc_K(M)}, C_1, C_2, C_3, \{C_{4k}, C'_{4k}\}_{\forall k\in \Gamma_{\Phi}}, \{C_{r^k_i}, C'_{r^k_i}\}_{\forall r^k_i\in \Gamma}, \Gamma, \Gamma_{\Phi}\right>$ for the plaintext message $\mathtt{M}$, where:
	\begin{align}
	C_1= & \mathtt{K}\cdot Y^d= \mathtt{K}\cdot \hat{e}(g, g)^{y\cdot d}\\
	C_2= & (h^k_1)^{d_j}= g^{\eta_k\cdot d_j}\\
	C_3= & (\mathtt{Pub^{2k}_c})^{d_j}= g^{\mathtt{x_k}\cdot \mathtt{Priv^k_c}\cdot d_j}\\
	C_{4k}= &\left(\mathtt{Pub^{1k}_c}\right)^{d_k}= g^{\mu_k\cdot \mathtt{Priv^k_c}\cdot d_k}\\
	C'_{4k}= &\left(\mathtt{Pub^{1k}_c}\right)^{d'_k}= g^{\mu_k\cdot \mathtt{Priv^k_c}\cdot d'_k}\\
	C_{r^k_i}= & \left(\mathtt{PK_{r^k_i}}\right)^{d_{r^k_i}\cdot H_1(w)}= g^{H_1(w)\cdot d_{r^k_i}\prod_{r^k_j\in \mathbb{R}_{r^k_i}}t_{r^k_j}}\\
	C'_{r^k_i}= & \left(\mathtt{PK_{r^k_i}}\right)^{d'_{r^k_i}\cdot H_1(w)}= g^{H_1(w)\cdot d'_{r^k_i}\prod_{r^k_j\in \mathbb{R}_{r^k_i}}t_{r^k_j}}
	\end{align}
	Note that the data owner embeds the hashed value of the keyword, i.e, $H_1(w)$ for some roles in $\Gamma$ only (this fixed position can be seen as part of the public parameter). 
	
	\subsubsection{Trapdoor Generation}
	\label{trapdoor_generation}
	In this phase, a user generates trapdoor $\mathtt{Trap}$ using his/her secret keys and the keywords of his/her choice for delegating keyword search capabilities to the public cloud. The user sends the trapdoor $\mathtt{Trap}$ along with the associated roles to the public cloud using a secure-channel. This phase comprises the \textsc{TrapGen} algorithm which is described next.
	\paragraph{\textsc{TrapGen} $(\left(\mathtt{Trap}, v\right)\leftarrow (\{\mathtt{RK^{1, u}_{r^k_x}}, \mathtt{RK^{2, u}_{r^k_x}}\}_{\forall r^k_x\in \mathbb{S}_{\mathtt{ID_u}}}, \mathtt{SK^k_{ID_u}}, \mathbb{S}_{\mathtt{ID_u}}, w))$} Suppose the user $\mathtt{ID_u}$ who possesses roles $\mathbb{S}_{\mathtt{ID_u}}$ wants to access the ciphertexts associated with the keyword $w$ of the $k^{th}$ organization. User $\mathtt{ID_u}$ chooses a random secret $v\in \mathbb{Z}_q^*$, current timestamp $ts$, and then he computes a trapdoor $\mathtt{Trap}=\left<tr_1, tr_2, tr_3, tr_4, \{tr^1_{r^k_x}, tr^2_{r^k_x}\}_{\forall r^k_x\in \mathbb{S}_{\mathtt{ID_u}}}, \mathbb{S}_{\mathtt{ID_u}}, ts\right>$, where:
	\begin{align}
	tr_1= & \left[\frac{\mathtt{Priv_{ID_u}}+ ts}{H_2\left(\mathtt{Priv^{k}_{ID_u}}\right)}\right] v= \frac{\left[\mathtt{Priv_{ID_u}}+ ts\right] v}{H_2\left(\mathtt{Priv^{k}_{ID_u}}\right)}\\
	tr_2= & \left(\mathtt{Priv^{k}_{ID_u}}\right)^{v}= g^{\frac{\left[y\cdot \mathtt{Priv_{ID_u}}+ \mathtt{x_k}\right]\cdot v}{\eta_k}}\\
	tr_3= & g^{\frac{v}{\mathtt{Priv_{ID_u}}}}\\
	tr_4= & g^v\\
	tr^1_{r^k_x}= & \left(\mathtt{RK^{1, u}_{r^k_x}}\right)^{\frac{v}{H_1(w)}}= g^{\frac{\left[y\cdot \mathtt{Priv_{ID_u}}+ \mu_k\right]v}{H_1(w)\cdot \prod_{r^k_j\in \mathbb{R}_{r^k_x}}t_{r^k_j}}}\\
	tr^2_{r^k_x}= & \left(\mathtt{RK^{2, u}_{r^k_x}}\right)^{\frac{v}{H_1(w)}}
	= g^{\frac{\left[y\cdot \mathtt{Priv_{ID_u}}+ \mu_k\right]v}{H_1(w)\cdot t_{r^k_x}}}
	\end{align}
	The user $\mathtt{ID_u}$ keeps the random secret $v$ in a secure place for decryption of the ciphertexts in Section \ref{decryption}.
	\subsubsection{Data Search}
	\label{search}
	In this phase, the public cloud performs a keyword search operation on the ciphertexts using the trapdoor received from the requested user $\mathtt{ID_u}$. This phase consists of the \textsc{Authentication}, \textsc{KeySearch} and \textsc{PartialDec} algorithms. In the \textsc{Authentication} algorithm, the public cloud authenticates the user and checks freshness of the keyword search request. In the \textsc{KeySearch} algorithm, the public cloud performs all the search related operation for finding the ciphertexts which have a matching keyword with the trapdoor received from the user $\mathtt{ID_u}$. This will be done if and only if the user is legitimate and the keyword search request is valid. In the \textsc{PartialDec} algorithm, the public cloud partially decrypts the ciphertexts and finally sends the partially decrypted ciphertexts to the user $\mathtt{ID_u}$. The details of these algorithms are given next.
	\paragraph{\textsc{Authentication} $\Big(\left(V^1_3/\perp\right)\leftarrow \Big(\{\mathtt{Priv^k_c}\}_{\mathtt{\forall k\in \Gamma_{\Phi}}}, \mathtt{Trap}, \mathtt{Pub^k_{ID_u}}, \mathtt{ID_u}, ts'\Big)\Big)$}
	\label{authetication}
	Before performing computationally expensive operations, the public cloud first authenticates the requesting user. During the authentication process, the public cloud also checks the freshness of search request by comparing the timestamp $ts$ associated with the trapdoor to its own current timestamp $ts'$ for preventing replay attacks. If the authentication fails or if the timestamp associated with the trapdoor represents a past time, the public cloud aborts the connection, i.e., returns $\perp$. Otherwise, it performs keyword search operations defined in the \textsc{KeySearch} algorithm. To authenticate the user and check the freshness of the request, the public cloud computes $U', V^1_1, V^2_1$ and $V^3_1$, based on its known  $\{\mathtt{Priv^k_c}\}_{\forall k\in \Gamma_{\Phi}}$) keys, where:
	\begin{align}
	U'= & \prod_{\forall k\in \Gamma_{\Phi}} \left(C'_{4k}\right)^{\frac{1}{\mathtt{Priv^k_c}}}\\
	= & \prod_{\forall k\in \Gamma_{\Phi}}\left(g^{\mu_k\cdot \mathtt{Priv^k_c}\cdot d_k'}\right)^{\frac{1}{\mathtt{Priv^k_c}}}\\
	= & g^{\sum_{\forall k\in \Gamma_{\Phi}}\mu_k\cdot d'_k}
	\end{align}
	\begin{align}
	V^1_1= & \hat{e}\left(\left(\mathtt{Pub^k_{ID_u}}\right)^{tr_1}, U'\right)\\
	= & \hat{e}\left(\left(g^{\frac{H_2\left(\mathtt{Priv^{k}_{ID_u}}\right)}{\mathtt{Priv_{ID_u}}}}\right)^{\frac{\left[\mathtt{Priv_{ID_u}}+ ts\right] v}{H_2\left(\mathtt{Priv^{k}_{ID_u}}\right)}}, g^{\sum_{\forall k\in \Gamma_{\Phi}}\mu_k\cdot d'_k}\right)\\
	= & \hat{e}\left(g^{v}, g^{\sum_{\forall k\in \Gamma_{\Phi}}\mu_k\cdot d'_k}\right)\cdot \hat{e}\left(g^{\frac{v\cdot ts}{\mathtt{Priv_{ID_u}}}}, g^{\sum_{\forall k\in \Gamma_{\Phi}}\mu_k\cdot d'_k}\right)\\
	= & \hat{e}\left(g, g\right)^{v\sum_{\forall k\in \Gamma_{\Phi}}\mu_k\cdot d'_k}\cdot \hat{e}\left(g, g\right)^{\frac{v\cdot ts \sum_{\forall k\in \Gamma_{\Phi}}\mu_k\cdot d'_k}{\mathtt{Priv_{ID_u}}}}
	\end{align}
	\begin{align}
	V^2_1= & \hat{e}\left((tr_3)^{ts}, U'\right)\\
	= & \hat{e}\left(g^{\frac{v\cdot ts}{\mathtt{Priv_{ID_u}}}}, g^{\sum_{\forall k\in \Gamma_{\Phi}}\mu_k\cdot d'_k}\right)\\
	= & \hat{e}\left(g, g\right)^{\frac{v\cdot ts\sum_{\forall k\in \Gamma_{\Phi}}\mu_k\cdot d'_k}{\mathtt{Priv_{ID_u}}}}
	\end{align}
	\begin{align}
	V^3_1= & \hat{e}\left(tr_4, U'\right) \\
	= & \hat{e}\left(g^v, g^{\sum_{\forall k\in \Gamma_{\Phi}}\mu_k\cdot d'_k}\right)\\
	= & \hat{e}\left(g, g\right)^{v\sum_{\forall k\in \Gamma_{\Phi}}\mu_k\cdot d'_k}
	\end{align}
	Now, the public cloud checks whether $V^1_1\stackrel{?}{=}V^2_1\cdot V^3_1$. If the equation holds, the public cloud performs the operations defined in the \textsc{KeySearch} algorithm. Otherwise, it aborts the connection. 
	\par 
	\emph{Proof of consistency}:
	\begin{align}
	V^1_1= & \hat{e}\left(g, g\right)^{v\sum_{\forall k\in \Gamma_{\Phi}}\mu_k\cdot d'_k}\cdot \hat{e}\left(g, g\right)^{\frac{v\cdot ts \sum_{\forall k\in \Gamma_{\Phi}}\mu_k\cdot d'_k}{\mathtt{Priv_{ID_u}}}}\\
	= & V^3_1\cdot V^2_1
	\end{align}
	
	\paragraph{\textsc{KeySearch} $((\mathbb{CT}/\perp)\leftarrow (\mathbb{CT}, \mathtt{Trap}, V^3_1))$} Suppose the user $\mathtt{ID_u}$ possesses a role set $\mathbb{S}_{\mathtt{ID_u}}$ and wants to access the $k^{th}$ organization's data. Suppose $\mathbb{CT}= \big<\mathtt{Enc_K(M)}, C_1, C_2, C_3, \{C_{4k}, C'_{4k}\}_{\forall k\in \Gamma_{\Phi}}, \{C_{r^k_i}, C'_{r^k_i}\}_{\forall r^k_i\in \Gamma}, \Gamma, \\\Gamma_{\Phi}\big>$ is the ciphertext of the $k^{th}$ organization on which the public cloud wants to perform the keyword search operation, where for all $r^k_i\in \Gamma$, there is at least one $r^k_x \in \mathbb{S}_{\mathtt{ID_u}}$ such that $r^k_x\in \mathbb{R}_{r^k_i}$.
	
	The public cloud computes $V^1_{r^k_x}$ and $V_2$. While computing $V^1_{r^k_x}$, two cases are considered which are as follows: 
	\par 
	Case 1: if $r^k_x== r^k_i$, then
	\begin{align}
	V^1_{r^k_x}= & \hat{e}\left(tr^1_{r^k_x}, C'_{r^k_i}\right)\\
	= & \hat{e}\left(g^{\frac{\left[y\cdot \mathtt{Priv_{ID_u}}+ \mu_k\right]v}{H_1(w)\cdot \prod_{r^k_j\in \mathbb{R}_{r^k_x}}t_{r^k_j}}}, g^{H_1(w)\cdot d'_{r^k_i}\prod_{r^k_j\in \mathbb{R}_{r^k_i}}t_{r^k_j}}\right)\\
	= & \hat{e}\left(g, g\right)^{\left[y\cdot \mathtt{Priv_{ID_u}}+ \mu_k\right]\cdot v\cdot d'_{r^k_i}}, \text{\Big(as $\mathbb{R}_{r^k_x}= \mathbb{R}_{r^k_i}$\Big)}\\
	= & \hat{e}\left(g, g\right)^{\left[y\cdot \mathtt{Priv_{ID_u}}\right]v\cdot d'_{r^k_i}}\cdot \hat{e}\left(g, g\right)^{\mu_k\cdot v\cdot d'_{r^k_i}}
	\end{align} 
	Otherwise, Case 2: if $r^k_x\in \left(\mathbb{R}_{r^k_i}\setminus \{r^k_i\}\right)$ (let $\gamma= \left[y\cdot \mathtt{Priv_{ID_u}}+ \mu_k\right]$)
	\begin{align}
	V^1_{r^k_x}= & \hat{e}\left(\left(tr^2_{r^k_x}\right)^{\mathtt{PKey^{r^k_x}_{r^k_i}}}, C'_{r^k_i}\right)\\
	= & \hat{e}\left(g^{\frac{\gamma \cdot v}{H_1(w)\cdot t_{r^k_x}}\cdot \frac{1}{\prod_{\forall r^k_j\in \mathbb{R}_{r^k_i}\setminus \{r^k_x\}} t_{r^k_j}}}, g^{H_1(w)\cdot d'_{r^k_i}\prod_{r^k_j\in \mathbb{R}_{r^k_i}}t_{r^k_j}}\right)\\
	= & \hat{e}\left(g^{\frac{\gamma\cdot v}{\prod_{\forall r^k_j\in \mathbb{R}_{r^k_i}}t_{r^k_j}}}, g^{d'_{r^k_i}\prod_{r^k_j\in \mathbb{R}_{r^k_i}}t_{r^k_j}}\right)\\
	= & \hat{e}\left(g, g\right)^{\gamma\cdot v\cdot d'_{r^k_i}}\\
	= & \hat{e}\left(g, g\right)^{\left[y\cdot \mathtt{Priv_{ID_u}}+ \mu_k\right] v\cdot d'_{r^k_i}}\\
	= & \hat{e}\left(g, g\right)^{\left[y\cdot \mathtt{Priv_{ID_u}}\right]v\cdot d'_{r^k_i}}\cdot \hat{e}\left(g, g\right)^{\mu_k\cdot v\cdot d'_{r^k_i}}
	\end{align}
	\begin{align}
	V_2= & \prod V^1_{r^k_x}\\
	= & \hat{e}\left(g, g\right)^{\left[y\cdot \mathtt{Priv_{ID_u}}\right]v\cdot \sum d'_{r^k_i}}\cdot \hat{e}\left(g, g\right)^{v\sum \mu_k \cdot d'_{r^k_i}}\\
	= & \hat{e}\left(g, g\right)^{\left[y\cdot \mathtt{Priv_{ID_u}}\right]v\cdot d_j}\cdot \hat{e}\left(g, g\right)^{v\sum \mu_k \cdot d'_{r^k_i}}
	\end{align}
	
	Now, the public cloud computes $V_3$, where
	\begin{align}\label{com1}
	V_3= & \frac{V_2}{V^3_1}= \frac{\hat{e}\left(g, g\right)^{\left[y\cdot \mathtt{Priv_{ID_u}}\right]v\cdot d_j}\cdot \hat{e}\left(g, g\right)^{v\sum \mu_k \cdot d'_{r^k_i}}}{\hat{e}\left(g, g\right)^{v\sum\mu_k\cdot d'_k}}\\
	= & \hat{e}\left(g, g\right)^{y\cdot \mathtt{Priv_{ID_u}}\cdot d_j\cdot v}
	\end{align}
	Note that $\hat{e}\left(g, g\right)^{v\sum \mu_k \cdot d_{r^k_i}}= \hat{e}\left(g, g\right)^{v\sum\mu_k\cdot d'_k}$ (Please refer Section \ref{encryption}).
	\par
	Afterward, the public cloud computes $V_4, V_5$ and $V_6$, where
	
	\begin{align}
	V_4= & \hat{e}\left(tr_2, C_2\right)\\
	= & \hat{e}\left(g^{\frac{\left[y\cdot \mathtt{Priv_{ID_u}}+ \mathtt{x_k}\right]\cdot v}{\eta_k}}, g^{\eta_k\cdot d_j}\right)\\
	= & \hat{e}\left(g, g\right)^{\left[y\cdot \mathtt{Priv_{ID_u}}+ \mathtt{x_k}\right]\cdot v\cdot d_j}\\
	= & \hat{e}\left(g, g\right)^{y\cdot \mathtt{Priv_{ID_u}}\cdot v\cdot d_j}\cdot \hat{e}\left(g, g\right)^{\mathtt{x_k}\cdot v\cdot d_j}
	\end{align}
	\begin{align}
	V_5= & \hat{e}\left((tr_4)^{\frac{1}{\mathtt{Priv^k_c}}}, C_3\right)\\
	= & \hat{e}\left(g^{\frac{v}{\mathtt{Priv^k_c}}}, g^{\mathtt{x_k}\cdot \mathtt{Priv^k_c}\cdot d_j}\right)\\
	= & \hat{e}\left(g, g\right)^{\mathtt{x_k}\cdot v\cdot d_j}
	\end{align}
	\begin{align}\label{com2}
	V_6= & \frac{V_1}{V_2}\\
	= & \frac{\hat{e}\left(g, g\right)^{y\cdot \mathtt{Priv_{ID_u}}\cdot v\cdot d_j}\cdot \hat{e}\left(g, g\right)^{\mathtt{x_k}\cdot v\cdot d_j}}{\hat{e}\left(g, g\right)^{\mathtt{x_k}\cdot v\cdot d_j}}\\
	= & \hat{e}\left(g, g\right)^{y\cdot \mathtt{Priv_{ID_u}}\cdot d_j\cdot v}
	\end{align}
	Finally, the public cloud compares the equations (\ref{com1}) and (\ref{com2}). If both are equal then it performs the operations defined in the \textsc{PartialDec} algorithm (described in Section \ref{partial_dec}). Otherwise, it aborts all the operations and outputs $\perp$, which means that the ciphertext does not have the desired keyword.
	
	\paragraph{\textsc{PartialDec} $\left(\mathbb{CT}'\leftarrow \left(\mathbb{CT}, \mathtt{Trap}, \{\mathtt{Priv^k_c}\}_{\forall k\in \Gamma_{\Phi}}, \mathbb{S}_{\mathtt{ID_u}}\right)\right)$}
	\label{partial_dec}
	In this algorithm, the public cloud partially decrypts all the ciphertexts returned by the \textsc{KeySearch} algorithm. Suppose ciphertext $\mathbb{CT}= \left<\mathtt{Enc_K(M)}, C_1, C_2, C_3, \{C_{4k}, C'_{4k}\}_{\forall k\in \Gamma_{\Phi}}, \{C_{r^k_i}, C'_{r^k_i}\}_{\forall r^k_i\in \Gamma}, \Gamma, \Gamma_{\Phi}\right>$ has a matching keyword with the trapdoor $\mathtt{Trap}$. To partially decrypt the ciphertext $\mathbb{CT}$, the public cloud first computes $V^7_{r^k_x}$ and $V_7$. Similar to $V^1_{r^k_x}$, the computation procedure considers the two following cases to compute $V^7_{r^k_x}$: 
	\par  
	Case 1: if $r^k_x== r^k_i$, then
	\begin{align}
	V^7_{r^k_x}= & \hat{e}\left(tr^1_{r^k_x}, C_{r^k_i}\right)\\
	= & \hat{e}\left(g^{\frac{\left[y\cdot \mathtt{Priv_{ID_u}}+ \mu_k\right]v}{H_1(w)\cdot \prod_{r^k_j\in \mathbb{R}_{r^k_x}}t_{r^k_j}}}, g^{H_1(w)\cdot d_{r^k_i}\prod_{r^k_j\in \mathbb{R}_{r^k_i}}t_{r^k_j}}\right)\\
	= & \hat{e}\left(g, g\right)^{\left[y\cdot \mathtt{Priv_{ID_u}}+ \mu_k\right]\cdot v\cdot d_{r^k_i}}, \text{\Big(as $\mathbb{R}_{r^k_x}= \mathbb{R}_{r^k_i}$\Big)}\\
	= & \hat{e}\left(g, g\right)^{\left[y\cdot \mathtt{Priv_{ID_u}}\right]v\cdot d_{r^k_i}}\cdot \hat{e}\left(g, g\right)^{\mu_k\cdot v\cdot d_{r^k_i}}
	\end{align} 
	Otherwise, Case 2: if $r^k_x\in \left(\mathbb{R}_{r^k_i}\setminus \{r^k_i\}\right)$ (let $\gamma= \left[y\cdot \mathtt{Priv_{ID_u}}+ \mu_k\right]$)
	\begin{align}
	V^7_{r^k_x}= & \hat{e}\left(\left(tr^2_{r^k_x}\right)^{\mathtt{PKey^{r^k_x}_{r^k_i}}}, C_{r^k_i}\right)\\
	= & \hat{e}\left(g^{\frac{\gamma \cdot v}{H_1(w)\cdot t_{r^k_x}}\cdot \frac{1}{\prod_{\forall r^k_j\in \mathbb{R}_{r^k_i}\setminus \{r^k_x\}} t_{r^k_j}}}, g^{H_1(w)\cdot d_{r^k_i}\prod_{r^k_j\in \mathbb{R}_{r^k_i}}t_{r^k_j}}\right)\\
	= & \hat{e}\left(g^{\frac{\gamma\cdot v}{\prod_{\forall r^k_j\in \mathbb{R}_{r^k_i}}t_{r^k_j}}}, g^{d_{r^k_i}\prod_{r^k_j\in \mathbb{R}_{r^k_i}}t_{r^k_j}}\right)\\
	= & \hat{e}\left(g, g\right)^{\gamma\cdot v\cdot d_{r^k_i}}\\
	= & \hat{e}\left(g, g\right)^{\left[y\cdot \mathtt{Priv_{ID_u}}+ \mu_k\right] v\cdot d_{r^k_i}}\\
	= & \hat{e}\left(g, g\right)^{\left[y\cdot \mathtt{Priv_{ID_u}}\right]v\cdot d_{r^k_i}}\cdot \hat{e}\left(g, g\right)^{\mu_k\cdot v\cdot d_{r^k_i}}
	\end{align}
	\begin{align}
	V_7= & \prod V^7_{r^k_x}\\
	= & \hat{e}\left(g, g\right)^{\left[y\cdot \mathtt{Priv_{ID_u}}\right]v\cdot \sum d_{r^k_i}}\cdot \hat{e}\left(g, g\right)^{v\sum \mu_k \cdot d_{r^k_i}}\\
	= & \hat{e}\left(g, g\right)^{\left[y\cdot \mathtt{Priv_{ID_u}}\right]v\cdot d_i}\cdot \hat{e}\left(g, g\right)^{v\sum \mu_k \cdot d_{r^k_i}}
	\end{align}
	The public cloud knowing its private key $\{\mathtt{Priv^k_c}\}_{\forall k\in \Gamma_{\Phi}}$ computes $U$ and $V_8$, as follows:
	\begin{align}
	U= & \prod_{\forall k\in \Gamma_{\Phi}} \left(C_{4k}\right)^{\frac{1}{\mathtt{Priv^k_c}}}\\
	= & \prod_{\forall k\in \Gamma_{\Phi}}\left(g^{\mu_k\cdot \mathtt{Priv^k_c}\cdot d_k}\right)^{\frac{1}{\mathtt{Priv^k_c}}}\\
	= & g^{\sum_{\forall k\in \Gamma_{\Phi}}\mu_k\cdot d_k}\\
	V_8= &
	\hat{e}\left(tr_4, U\right)\\
	= & \hat{e}\left(g^v, g^{\sum_{\forall k\in \Gamma_{\Phi}}\mu_k\cdot d_k}\right)\\
	= & \hat{e}\left(g, g\right)^{v\sum_{\forall k\in \Gamma_{\Phi}}\mu_k\cdot d_k}
	\end{align}
	Now, the public cloud computes $V_9$, where:
	\begin{align}
	V_9= & \frac{V_7}{V_8}= \frac{\hat{e}\left(g, g\right)^{\left[y\cdot \mathtt{Priv_{ID_u}}\right]v\cdot d_i}\cdot \hat{e}\left(g, g\right)^{v\sum \mu_k \cdot d_{r^k_i}}}{\hat{e}\left(g, g\right)^{v\sum\mu_k\cdot d_k}}\\
	= & \hat{e}\left(g, g\right)^{y\cdot \mathtt{Priv_{ID_u}}\cdot d_i\cdot v}
	\end{align}
	Note that $\hat{e}\left(g, g\right)^{v\sum \mu_k \cdot d_{r^k_i}}= \hat{e}\left(g, g\right)^{v\sum\mu_k\cdot d_k}$ (Please refer Section \ref{encryption}).
	
	The public cloud computes $V_9$, where:
	\begin{align}
	V_{10}= & V_6\cdot V_9\\
	= & \hat{e}\left(g, g\right)^{y\cdot \mathtt{Priv_{ID_u}}\cdot d_j\cdot v}\cdot \hat{e}\left(g, g\right)^{y\cdot \mathtt{Priv_{ID_u}}\cdot d_i\cdot v}\\
	= & \hat{e}\left(g, g\right)^{y\cdot \mathtt{Priv_{ID_u}}\cdot v\left[d_j+ d_i\right]}\\
	= & \hat{e}\left(g, g\right)^{y\cdot \mathtt{Priv_{ID_u}}\cdot v\cdot d}
	\end{align}
	Finally, the public cloud sends the partially decrypted ciphertext $\mathbb{CT}'= \left<\mathtt{Enc_K(M)}, C_1, V_{10}\right>$ to the user $\mathtt{ID_u}$.
	
	\subsubsection{Decryption}
	\label{decryption}
	In this phase, the user $\mathtt{ID_u}$ decrypts the received partially decrypted ciphertext $\mathbb{CT}'$ using his/her private key $\mathtt{Priv_{ID_u}}$ and random secret $v$. This phase comprises the \textsc{FullDEC} algorithm which is described next.
	\paragraph{\textsc{FullDEC}$(\mathtt{M}\leftarrow (\mathbb{CT}', \mathtt{Priv_{ID_u}}, v))$} It computes $\mathtt{K}$ from the ciphertext $\mathbb{CT}'$ using his/her secret keys, $\mathtt{priv_{ID_u}}$ and $v$.
	\begin{align}
	\mathtt{K}= &\frac{C_1}{\left(V_{10}\right)^{\frac{1}{\mathtt{Priv_{ID_u}}\cdot v}}}\\
	= & \frac{\mathtt{K}\cdot \hat{e}\left(g, g\right)^{y\cdot d}}{\left(\hat{e}\left(g, g\right)^{y\cdot \mathtt{Priv_{ID_u}}\cdot v\cdot d}\right)^{\frac{1}{\mathtt{Priv_{ID_u}}\cdot v}}}\\
	= & \frac{\mathtt{K}\cdot \hat{e}\left(g, g\right)^{y\cdot d}}{\hat{e}\left(g, g\right)^{y\cdot d}}
	\end{align}
	Finally, user $\mathtt{ID_u}$ gets the actual plaintext data by decrypting $\mathtt{Enc_{K}(M)}$ using $\mathtt{K}$ and removes the random secret $v$ from his/her database.
	
	\subsection{Conjunctive Keyword Search}
	\label{conjuctive_keyword_search}
	Many times a user wants to perform multiple keyword search using a single search request instead of sending multiple single keyword search requests. This property is called the \emph{Conjunctive Keyword Search}. The proposed scheme can provide conjunctive keyword search with the following modifications. The owner computes modified ciphertext components $C_{r^k_i}= \left(\mathtt{PK_{r^k_i}}\right)^{d_{r^k_i}\cdot \prod H_1(w_i)}= g^{d_{r^k_i}\prod H_1(w_i)\cdot \prod_{r^k_j\in \mathbb{R}_{r^k_i}}t_{r^k_j}}$ and $C'_{r^k_i}=  \left(\mathtt{PK_{r^k_i}}\right)^{d'_{r^k_i}\cdot \prod H_1(w_i)}= g^{d'_{r^k_i}\prod H_1(w_i)\cdot \prod_{r^k_j\in \mathbb{R}_{r^k_i}}t_{r^k_j}}$. Similarly, a user computes trapdoor components $tr^1_{r^k_x}= \left(\mathtt{RK^{1, u}_{r^k_x}}\right)^{\frac{v}{\prod H_1(w_i)}}$ and $tr^2_{r^k_x}= \left(\mathtt{RK^{2, u}_{r^k_x}}\right)^{\frac{v}{\prod H_1(w_i)}}$. It can be observed that, our conjunctive keyword search mechanism does not introduce any additional overhead in the system.

	\subsection{Revocation}
	\label{revocation}
	In the proposed scheme, a SA can revoke a user in two ways, namely \emph{complete user revocation} and \emph{role-level revocation}. The former revocation method means that the user can no longer access any data belonging to that organization. The later revocation method represents that if one or more roles of a user is revoked, the user can still access data with his/her non-revoked roles if they are qualified enough according to the RBAC access policy.
	\par  
	The complete user revocation is achieved by revoking the public key $\mathtt{Pub^k_{ID_u}}$ of the user, so that the public cloud do not use it during the authentication process in the \textsc{Authentication} algorithm defined in Section \ref{authetication}. To do that, SA removes the pubic key $\mathtt{Pub^k_{ID_u}}$ of the revoked user $\mathtt{ID_u}$ from its public bulletin board, which can be done easily.
	\par 
	For the role-level revocation, the SA updates all the parameters related with the revoked role. Suppose the SA wants to revoke a role $r^k_i$ from one or more users. To do that, the SA first chooses a fresh random number $t'_{r^k_i}\in \mathbb{Z}_q^*$ and updates all the parameters related with the revoked role $r^k_i$. The SA computes updated public keys $\left(\mathtt{PK_{r^k_j}}\right)^{\frac{t'_{r^k_i}}{t_{r^k_i}}}$, role secrets $\left(\mathtt{RS_{r^k_j}}\cdot \frac{t'_{r^k_i}}{t_{r^k_i}}\right)$ and proxy re-encryption keys $\left(\mathtt{PKey^{r^k_w}_{r^k_j}}\cdot \frac{t'_{r^k_i}}{t_{r^k_i}}\right)$ related with the revoked role $r^k_i$ (i.e., for all $r^k_j$ such that $r^k_i\in \mathbb{R}_{r^k_j}$), where $t_{r^k_i}$ is the previously chosen random number associated with $r^k_i$. The SA then sends the $\frac{t'_{r^k_i}}{t_{r^k_i}}$ to the public cloud for re-encryption of the stored ciphertexts associated with the revoked role $r^k_i$. It also sends $\frac{t'_{r^k_i}}{t_{r^k_i}}$ to the corresponding role-managers for updating the role-keys associated with the revoked role $r^k_i$.
	\par 
	The public cloud re-encrypts the ciphertext components $\left(C_{r^k_j}\right)^{\frac{t'_{r^k_i}}{t_{r^k_i}}}$ and $\left(C'_{r^k_j}\right)^{\frac{t'_{r^k_i}}{t_{r^k_i}}}$ for all $r^k_j$ such that $r^k_i\in \mathbb{R}_{r^k_j}$. This is essential to prevent the revoked users from accessing the data using the revoked role (i.e., \emph{Backward Secrecy}).
	\par 
	Moreover, to enable the other non-revoked users for accessing the re-encrypted ciphertexts, the concerned role-managers need to send updated role-keys to the non-revoked users (i.e., \emph{Forward Secrecy}). The updated role-keys are computed as follows: i) $\left(\mathtt{RK^{1, u}_{r^k_i}}\right)^{\frac{t_{r^k_i}}{t'_{r^k_i}}}$ for all the non-revoked users who possess $r^k_i$ and ii) $\left(\mathtt{RK^{2, u}_{r^k_j}}\right)^{\frac{t_{r^k_i}}{t'_{r^k_i}}}$ for all the non-revoked users who possess $r^k_j$, such that $r^k_i\in \mathbb{R}_{r^k_j}$.

	\section{Analysis}
	\label{analysis}
	This section first presents security analysis of the proposed scheme, followed by its performance analysis. In the security analysis, we demonstrate that the proposed scheme is secure against chosen plaintext and chosen keyword attacks. In the performance analysis, we present a comprehensive performance analysis of the proposed scheme along with its experimental results.
	\subsection{Security Analysis}
	\label{security_analysis}
	\subsubsection{Security against Chosen Plaintext Attack}
	CPA security of the proposed scheme can be defined by the following theorem and proof.
	\begin{theorem}\label{theo1}
		If a probabilistic-polynomial time (PPT) adversary $\mathcal{A}_1$ wins the CPA security game as defined in Section \ref{CPA} with a non-negligible advantage $\epsilon$, then a PPT simulator $\mathcal{B}$ can be constructed to break the DBDH assumption with non-negligible advantage $\frac{\epsilon}{2}$.
	\end{theorem}
	\begin{proof}\label{cpa_proof}
		In this proof, we show that a simulator $\mathcal{B}$ can be constructed to help an adversary $\mathcal{A}_1$ to gain advantage $\frac{\epsilon}{2}$ against our proposed scheme.
		\par 
		The DBDH challenger $\mathcal{C}$ chooses random numbers $(a, b, c, z)\in \mathbb{Z}_q^*$ and flips a binary random coin $l$. It sets $Z= \hat{e}\left(g, g\right)^{abc}$ if $l= 0$ and $Z= \hat{e}\left(g, g\right)^z$ otherwise. Afterwards, challenger $\mathcal{C}$ sends $A= g^a, B= g^b, C= g^c$ and $Z$ to the simulator $\mathcal{B}$, and it asks the simulator $\mathcal{B}$ to output $l$. Now simulator $\mathcal{B}$ acts as a challenger in the rest of the security game.
		\par 
		In the following game, simulator $\mathcal{B}$ interacts with the adversary $\mathcal{A}_1$ as follows:
		\par
		\textbf{\textsc{Init}} Adversary $\mathcal{A}_1$ sends a challenged role set $\Gamma^*$, a keyword $w$ and two identities $(\mathtt{ID_u^*}, \mathtt{ID^*_c})$ to the simulator $\mathcal{B}$.
		\par 
		\textbf{\textsc{Setup}} Simulator $\mathcal{B}$ chooses random numbers $\{\zeta_k, \vartheta_k, \varrho_k\}_{\forall k\in \Phi}\in \mathbb{Z}_q^*$. It also chooses random numbers $\{\alpha_{r^k_i}\}_{\forall i\in \Psi_k, \forall k\in \Phi}\in \mathbb{Z}_q^*$. Simulator $\mathcal{B}$ computes $Y= \hat{e}\left(g, g\right)^{ab}= \hat{e}\left(A, B\right), \{h_1^k= g^{b\cdot \zeta_k}= B^{\zeta_k}\}_{\forall k\in \Phi}$. Simulator $\mathcal{B}$ also computes $\mathtt{PK_{r^k_i}}= g^{b\prod_{\forall r^k_j\in \mathbb{R}_{r^k_i}}\alpha_{r^k_j}}= B^{\prod_{\forall r^k_j\in \mathbb{R}_{r^k_i}}\alpha_{r^k_j}}$ for all $r^k_i\in \Psi_k$, where $1\le k\le m$. Moreover, simulator $\mathcal{B}$ computes $\left\{\left\{\mathtt{PKey^{r^k_w}_{r^k_i}}= \prod_{\forall r^k_j\in \mathbb{R}_{r^k_i}\setminus \{r^k_w\}}\alpha_{r^k_j}\right\}_{\forall r^k_w\in \mathbb{R}_{r^k_i}\setminus \{r^k_i\}}\right\}_{\forall r^k_i\in \Psi_k}$ where $1\le k\le m$.
		\par 
		Simulator $\mathcal{B}$ also chooses a random number $\mathtt{s_{ID^*_u}}\in \mathbb{Z}_q^*$ and computes $h_{id^*_c}= H_1(\mathtt{ID^*_c})$. It then computes $\{\mathtt{Priv^k_c}, \mathtt{Pub^{1k}_c}, \mathtt{Pub^{2k}_c}, \mathtt{Priv^{k}_{ID_u}}, \mathtt{Pub^k_{ID_u}}\}_{\forall k\in \Phi}$ and $\mathtt{Priv_{ID_u}}$, where 
		\begin{align}
		\mathtt{Priv^k_c}=& H_2\big(g^{\frac{a\cdot b\cdot h_{id^*_c}}{b\cdot \varrho_k}}\big)= H_2\big(A^{\frac{h_{id^*_c}}{\varrho_k}}\big)\\
		\mathtt{Pub^{1k}_c}=& g^{b\cdot \vartheta_k\cdot \mathtt{Priv^k_c}}= B^{\vartheta_k\cdot H_2\big(A^{\frac{h_{id^*_c}}{\varrho_k}}\big)}\\
		\mathtt{Pub^{2k}_c}=& g^{b\cdot \varrho_k\cdot \mathtt{Priv^k_c}}= B^{\varrho_k\cdot H_2\big(A^{\frac{h_{id^*_c}}{\varrho_k}}\big)} \\
		\mathtt{Priv^{k}_{ID_u}}=& g^{\frac{a\cdot b\cdot \mathtt{s_{ID^*_u}}+ b\cdot \varrho_k}{b\cdot \zeta_k}}= A^{\frac{\mathtt{s_{ID^*_u}}}{\zeta_k}}\cdot g^{\frac{\varrho_k}{\zeta_k}} \\
		\mathtt{Pub^k_{ID_u}}=& g^{\frac{H_2\left(\mathtt{Priv^{k}_{ID_u}}\right)}{\mathtt{s_{ID^*_u}}}}=g^{\frac{H_2\left(A^{\frac{\mathtt{s_{ID^*_u}}}{\zeta_k}}\cdot g^{\frac{\varrho_k}{\zeta_k}}\right)}{\mathtt{s_{ID^*_u}}}}\\
		\mathtt{Priv_{ID_u}}=& \mathtt{s_{ID^*_u}}
		\end{align}
		Finally, simulator $\mathcal{B}$ sends the following parameters to the adversary $\mathcal{A}_1$: $\big<q, \mathbb{G}_1, \mathbb{G}_T, \hat{e}, H_1, H_2, Y, \{h_1^k\}_{\forall k\in \Phi}, \{\{\mathtt{PK_{r^k_i}}\}_{\forall r^k_i\in \Psi_k}\}_{\forall k\in \Phi},\\ \Big\{\Big\{\mathtt{PKey^{r^k_w}_{r^k_i}}\Big\}_{\forall r^k_w\in \mathbb{R}_{r^k_i}\setminus \{r^k_i\}}\Big\}_{\forall r^k_i\in \Psi_k, \forall k\in \Phi}\big>$. Simulator $\mathcal{B}$ also sends $\{\mathtt{Priv^k_c}, \mathtt{Pub^{1k}_c}, \mathtt{Pub^{2k}_c}, \mathtt{Priv^{k}_{ID_u}}, \mathtt{Pub^k_{ID_u}}\}_{\forall k\in \Phi}$ and $\mathtt{Priv_{ID_u}}$ to the adversary $\mathcal{A}$. Note that simulator $\mathcal{B}$ sends a random number $\mathtt{s_{ID^*_u}}$ as $\mathtt{Priv_{ID_u}}$ to the adversary $\mathcal{A}$. As the simulator $\mathcal{B}$ chooses $\mathtt{s_{ID^*_u}}$ in the \textsc{Setup} and sends it to the adversary $\mathcal{A}$, the simulated game remains the same as the original scheme.
		
		\par
		\textbf{\textsc{Phase 1}} Adversary sends a challenged role set $\mathbb{S}^*$ to the simulator $\mathcal{B}$ for role-keys. Simulator $\mathcal{B}$ computes $\{\mathtt{RK^{1, u}_{r^k_x}}, \mathtt{RK^{2, u}_{r^k_x}}\}_{\forall r^k_x\in \mathbb{S}^*}$ as follows:
		\par 
		For all $r^k_x\in \mathbb{S}^*$, simulator $\mathcal{B}$ computes 
		\begin{align}
		\mathtt{RK^{1, u}_{r^k_x}}=& 	g^{\frac{a\cdot b\cdot \mathtt{Priv_{ID_u}}+ b\cdot \vartheta_k}{b\prod_{\forall r^k_j\in \mathbb{R}_{r^k_x}}\alpha_{r^k_j}}}= A^{\frac{H_2\big(A^{h_{id^*_u}}\big)}{\prod_{\forall r^k_j\in \mathbb{R}_{r^k_x}}\alpha_{r^k_j}}}\cdot g^{\frac{\vartheta_k}{\prod_{\forall r^k_j\in \mathbb{R}_{r^k_x}}\alpha_{r^k_j}}}\\
		\mathtt{RK^{2, u}_{r^k_x}}=& g^{\frac{a\cdot b\cdot \mathtt{Priv_{ID_u}}+ b\cdot \vartheta_k}{b\cdot \alpha_{r^k_x}}}= A^{\frac{H_2\big(A^{h_{id^*_u}}\big)}{\alpha_{r^k_x}}}\cdot g^{\frac{\vartheta_k}{\alpha_{r^k_x}}}
		\end{align}
		Finally, simulator $\mathcal{B}$ sends $\{\mathtt{RK^{1, u}_{r^k_x}}, \mathtt{RK^{2, u}_{r^k_x}}\}_{\forall r^k_x\in \mathbb{S}^*}$ to the adversary $\mathcal{A}_1$. Note that distribution of the role-keys for $\mathbb{S}^*$ is identical to the original scheme. 
		\par 
		\textbf{\textsc{Challenge}} When adversary $\mathcal{A}_1$ decides  that \textbf{\textsc{Phase 1}} is over, it submits two equal length messages $\mathtt{K_0}$ and $\mathtt{K_1}$ to the simulator $\mathcal{B}$. Simulator $\mathcal{B}$ flips a random binary coin $\omega$ and encrypts $\mathtt{K_\omega}$ with the challenged role set $\Gamma^*$.
		\par 
		Simulator $\mathcal{B}$ first computes $h_{w}= H_1(w)$ and chooses five polynomials $q_1(x), q_2(x), q_3(x), q_4(x)$ and $q_5(x)$ of degree $2, |\Gamma^*_{\Phi}|, |\Gamma^*_{\Phi}|, |\Gamma^*|$ and $|\Gamma^*|$ respectively, where $\Gamma^*_{\Phi}$ represents the set of system authorities associated with $\Gamma^*$, as follows:
		
		\begin{itemize}
			\item $q_1(x)$: Simulator $\mathcal{B}$ implicitly sets $q_1(0)= c$ and randomly chooses the rest of the points to define the polynomial $q_1(x)$ completely. Note that $q_1(1)$ and $q_1(2)$ values implicitly represent $d_i$ and $d_j$ of our original scheme respectively.
			\item $q_2(x)$: Simulator $\mathcal{B}$ sets $q_2(0)= q_1(1)$ and randomly chooses the rest of the  points to define $q_2(x)$ completely.
			\item $q_3(x)$: Simulator $\mathcal{B}$ sets $q_3(0)= q_1(2)$ and randomly chooses the rest of the points to defined $q_3(x)$ completely.
			\item $q_4(x)$: Simulator $\mathcal{B}$ sets $q_4(0)= q_1(1)$ and randomly chooses the rest of the points to define $q_4(x)$ completely.
			\item $q_5(x)$: Simulator $\mathcal{B}$ sets $q_5(0)= q_1(2)$ and randomly chooses the rest of the points to define $q_5(x)$ completely.
		\end{itemize}
		\par 
		Now, simulator $\mathcal{B}$ computes a challenged ciphertext $\mathbb{CT}_\omega=\big<C_1, C_2, C_3, \{C_{4k}, C'_{4k}\}_{\forall k\in \Gamma^*_{\Phi}}, \{C_{r^k_i}, C'_{r^k_i}\}_{\forall r^k_i\in \Gamma^*}\big>$, where
		\begin{align}
		C_1=& \mathtt{K_\omega}\cdot Z\\
		C_2=& g^{b\cdot \zeta_k\cdot q_1(2)}= B^{\zeta_k\cdot q_1(2)}\\
		C_3=& g^{b\cdot \varrho_k\cdot \mathtt{Priv^k_c}\cdot q_1(2)}= B^{\varrho_k\cdot H_2\big(A^{\frac{h_{id^*_c}}{\varrho_k}}\big)\cdot q_1(2)}\\
		C_{4k}=& g^{b\cdot \vartheta_k\cdot \mathtt{Priv^k_c}\cdot q_2(i)}\\
		=& B^{\vartheta_k\cdot H_2\big(A^{\frac{h_{id^*_c}}{\varrho_k}}\big)\cdot q_2(i)}, \text{ $1\le i\le |\Gamma^*_{\Phi}|$}\\
		C'_{4k}=& g^{b\cdot \vartheta_k\cdot \mathtt{Priv^k_c}\cdot q_3(i)}\\
		=& B^{\vartheta_k\cdot H_2\big(A^{\frac{h_{id^*_c}}{\varrho_k}}\big)\cdot q_3(i)}, \text{ $1\le i\le |\Gamma^*_{\Phi}|$}\\
		C_{r^k_i}=& g^{h_{w}\cdot b\cdot q_4(i)\prod_{\forall r^k_j\in \mathbb{R}_{r^k_i}}\alpha_{r^k_j}}, \text{ $1\le i\le |\Gamma^*|$}\\
		=& B^{h_{w}\cdot q_4(i)\prod_{\forall r^k_j\in \mathbb{R}_{r^k_i}}\alpha_{r^k_j}}\\
		C'_{r^k_i}=& g^{h_{w}\cdot b\cdot q_5(i)\prod_{\forall r^k_j\in \mathbb{R}_{r^k_i}}\alpha_{r^k_j}}, \text{ $1\le i\le |\Gamma^*|$}\\
		=& B^{h_{w}\cdot q_5(i)\prod_{\forall r^k_j\in \mathbb{R}_{r^k_i}}\alpha_{r^k_j}}
		\end{align}
		\par
		
		Note that $c$ (implicitly) can be recovered using the Lagrange's polynomial interpolation from the values $q_1(1)$ and $q_1(2)$, and $q_1(1)$, $q_1(2)$ can be recovered from the polynomials $q_4(x)$ and $q_5(x)$  if and only if the entity (i.e., adversary $\mathcal{A}_1$) possesses a qualified set of roles. Hence, the distribution of the ciphertext $\mathbb{CT}_\omega$ for $\Gamma^*$ is identical to the original scheme.
		\par 
		\textbf{\textsc{Phase 2}} Same as \textbf{\textsc{Phase 1}}
		
		\par 
		\textbf{\textsc{Guess}} The adversary $\mathcal{A}_1$ guesses a bit $\omega'$ which is sent to simulator $\mathcal{B}$. If $\omega'=\omega$ then the adversary $\mathcal{A}_1$ wins CPA game; otherwise it fails. If $\omega'= \omega$, simulator $\mathcal{B}$ answers ``DBDH'' in the game (i.e. outputs $l= 0$); otherwise $\mathcal{B}$ answers ``random'' (i.e. outputs $l= 1$).
		\par 
		If $Z= \hat{e}(g, g)^{z}$; then $C_{1}$ is completely random from the view of the adversary $\mathcal{A}_1$. So, the received ciphertext $\mathbb{CT}_\omega$ is not compliant to the game (i.e. invalid ciphertext). Therefore, the adversary $\mathcal{A}_1$ chooses $\omega'$ randomly. Hence, the probability of the adversary $\mathcal{A}_1$ for outputting $\omega'= \omega$ is $\frac{1}{2}$.
		\par 
		
		If $Z= \hat{e}(g, g)^{abc}$, then adversary $\mathcal{A}_1$ receives a valid ciphertext. The adversary $\mathcal{A}_1$ wins the CPA game with non-negligible advantage $\epsilon$ (according to Theorem \ref{theo1}). As such, the probability of outputting $\omega'= \omega$ for the adversary $\mathcal{A}_1$ is $\frac{1}{2}+ \epsilon$, where probability $\epsilon$ is for guessing that the received ciphertext is valid and probability $\frac{1}{2}$ is for guessing whether the valid encrypted message $C_{1}$ is related to $\mathtt{K_0}$ or $\mathtt{K_1}$.
		\par 
		Therefore, the overall advantage $Adv^{IND-CPA}_{\mathcal{A}_1}$ of the simulator $\mathcal{B}$ is $\frac{1}{2}(\frac{1}{2}+ \epsilon+ \frac{1}{2})- \frac{1}{2}= \frac{\epsilon}{2}$.
	\end{proof}
	\subsubsection{Security against Chosen Keyword Attack}
	\label{cka}
	Chosen keyword attack (CKA) security of the proposed scheme can be defined by the following theorem and proof. 
	\begin{theorem}\label{theo2}
		If a PPT adversary $\mathcal{A}_2$ wins the CKA security game defined in Section \ref{CKA} with a non-negligible advantage $\epsilon$, then a PPT simulator $\mathcal{B}$ can be constructed to break DBDH assumption with non-negligible advantage $\frac{\epsilon}{2}$.
	\end{theorem}
	\begin{proof}
		In this proof, we show that a simulator $\mathcal{B}$ can be constructed to help an adversary $\mathcal{A}_2$ to gain advantage $\frac{\epsilon}{2}$ against our proposed scheme.
		\par 
		The DBDH challenger $\mathcal{C}$ chooses random numbers $(a, b, c, z)\in \mathbb{Z}_q^*$ and flips a binary random coin $l$. It sets $Z= \hat{e}\left(g, g\right)^{abc}$ if $l= 0$ and $Z= \hat{e}\left(g, g\right)^z$ otherwise. Afterwards, challenger $\mathcal{C}$ sends $A= g^a, B= g^b, C= g^c$ and $Z$ to the simulator $\mathcal{B}$, and it asks the simulator $\mathcal{B}$ to output $l$. Now simulator $\mathcal{B}$ acts as a challenger in the rest of the security game.
		\par 
		In the following game simulator $\mathcal{B}$ interacts with the adversary $\mathcal{A}_2$ as follows:
		\par
		\textbf{\textsc{Init}} Adversary $\mathcal{A}_2$ sends a challenged role set $\Gamma^*$ and two identities $(\mathtt{ID^*_c}, \mathtt{ID_u^*})$ to the simulator $\mathcal{B}$.
		\par 
		\textbf{\textsc{Setup}} Simulator $\mathcal{B}$ chooses random numbers $\{\zeta_k, \vartheta_k, \varrho_k\}_{\forall k\in \Phi}$. It also chooses random numbers $\{\alpha_{r^k_i}\}_{\forall i\in \Psi_k, \forall k\in \Phi}\in \mathbb{Z}_q^*$. Simulator $\mathcal{B}$ computes $Y= \hat{e}\left(g, g\right)^{ab}= \hat{e}\left(A, B\right), \{h_1^k= g^{b\cdot \zeta_k}= B^{\zeta_k}\}_{\forall k\in \Phi}$. It also computes $\mathtt{PK_{r^k_i}}= g^{b\prod_{\forall r^k_j\in \mathbb{R}_{r^k_i}}\alpha_{r^k_j}}= B^{\prod_{\forall r^k_j\in \mathbb{R}_{r^k_i}}\alpha_{r^k_j}}$ for all $r^k_i\in \Psi_k$, where $1\le k\le m$. Moreover, simulator $\mathcal{B}$ computes $\left\{\left\{\mathtt{PKey^{r^k_w}_{r^k_i}}= \prod_{\forall r^k_j\in \mathbb{R}_{r^k_i}\setminus \{r^k_w\}}\alpha_{r^k_j}\right\}_{\forall r^k_w\in \mathbb{R}_{r^k_i}\setminus \{r^k_i\}}\right\}_{\forall r^k_i\in \Psi_k}$ where $1\le k\le m$.
		\par 
		Moreover, simulator $\mathcal{B}$ chooses a random number $\mathtt{s_{ID^*_u}}\in \mathbb{Z}_q^*$ and computes $h_{id^*_c}= H_1(\mathtt{ID^*_c})$. It then computes $\{\mathtt{Priv^k_c}, \mathtt{Pub^{1k}_c}, \mathtt{Pub^{2k}_c}, \mathtt{Pub^k_{ID_u}}\}_{\forall k\in \Phi}$, where
		\begin{align}
		\mathtt{Priv^k_c}=& H_2\big(g^{\frac{a\cdot b\cdot h_{id^*_c}}{b\cdot \varrho_k}}\big)= H_2\big(A^{\frac{h_{id^*_c}}{\varrho_k}}\big)\\
		\mathtt{Pub^{1k}_c}=& g^{b\cdot \vartheta_k\cdot \mathtt{Priv^k_c}}= B^{\vartheta_k\cdot H_2\big(A^{\frac{h_{id^*_c}}{\varrho_k}}\big)}\\
		\mathtt{Pub^{2k}_c}=& g^{b\cdot \varrho_k\cdot \mathtt{Priv^k_c}}= B^{\varrho_k\cdot H_2\big(A^{\frac{h_{id^*_c}}{\varrho_k}}\big)}\\
		\mathtt{Pub^k_{ID_u}}=& g^{\frac{H_2\left(g^{\frac{a\cdot b\cdot \mathtt{s_{ID^*_u}}+ b\cdot \varrho_k}{b\cdot \zeta_k}}\right)}{\mathtt{s_{ID^*_u}}}}= g^{\frac{H_2\left(A^{\frac{\mathtt{s_{ID^*_u}}}{\zeta_k}}\cdot g^{\frac{\varrho_k}{\zeta_k}}\right)}{\mathtt{s_{ID^*_u}}}}
		\end{align}
		\par 
		Simulator $\mathcal{B}$ sends the following parameters to the adversary $\mathcal{A}_2$: $\big<q, \mathbb{G}_1, \mathbb{G}_T, \hat{e}, H_1, H_2, Y, \{h_1^k\}_{\forall k\in \Phi}, \{\{\mathtt{PK_{r^k_i}}\}_{\forall r^k_i\in \Psi_k}\}_{\forall k\in \Phi},\\ \Big\{\Big\{\mathtt{PKey^{r^k_w}_{r^k_i}}\Big\}_{\forall r^k_w\in \mathbb{R}_{r^k_i}\setminus \{r^k_i\}}\Big\}_{\forall r^k_i\in \Psi_k, \forall k\in \Phi}, \{\mathtt{Pub^{1k}_c}, \mathtt{Pub^{2k}_c}, \\\mathtt{Pub^k_{ID_u}}\}_{\forall k\in \Phi}\big>$. Simulator $\mathcal{B}$ also sends private keys $\{\mathtt{Priv^k_c}\}_{\forall k\in \Phi}$ and $\mathtt{Priv_{ID_u}}= \mathtt{s_{ID^*_u}}$ to the adversary $\mathcal{A}_2$.
		\par
		\textbf{\textsc{Phase 1}} Adversary $\mathcal{A}_2$ sends a set of roles $\mathbb{S}^*$ and a keyword $w$ to the simulator $\mathcal{B}$ for the trapdoor. Simulator $\mathcal{B}$ chooses random numbers $(\mathbbm{v}, \mathbbm{ts})\in \mathbb{Z}_q^*$. It computes $h_w= H_1(w)$. Simulator $\mathcal{B}$ computes the trapdoor $\mathtt{Trap}= \big<tr_1, tr_2, tr_3, tr_4, \{tr^1_{r^k_x}, tr^2_{r^k_x}\}_{\forall r^k_x\in \mathbb{S}^*}\big>$, where
		\begin{align}
		tr_1=& \frac{\mathtt{s_{ID^*_u}}+ \mathbbm{ts}}{H_2\left(g^{\frac{a\cdot b\cdot \mathtt{s_{ID^*_u}}+ b\cdot \varrho_k}{b\cdot\zeta_k}}\right)}\cdot \mathbbm{v}= \frac{\left[\mathtt{s_{ID^*_u}}+ \mathbbm{ts}\right]\mathbbm{v}}{H_2\left(A^{\frac{\mathtt{s_{ID^*_u}}}{\zeta_k}}\cdot g^{\frac{\varrho_k}{\zeta_k}}\right)} \\
		tr_2=& \left(g^{\frac{a\cdot b\cdot \mathtt{s_{ID^*_u}}+ b\cdot \varrho_k}{b\cdot\zeta_k}}\right)^{\mathbbm{v}}= A^{\frac{\mathtt{s_{ID^*_u}\cdot \mathbbm{v}}}{\zeta_k}}\cdot g^{\frac{\varrho_k\cdot \mathbbm{v}}{\zeta_k}}  \\
		tr_3=& g^{\frac{\mathbbm{v}}{\mathtt{s_{ID^*_u}}}}= g^{\frac{\mathbbm{v}}{\mathtt{s_{ID^*_u}}}}  \\
		tr_4=& g^{\mathbbm{v}}  
		\end{align}
		For all $r^k_x\in \mathbb{S}^*$,
		\begin{align}
		tr^1_{r^k_x}=& 
		\left(g^{\frac{a\cdot b\cdot \mathtt{s_{ID^*_u}}+ b\cdot \vartheta_k}{b\prod_{\forall r^k_j\in \mathbb{R}_{r^k_x}}\alpha_{r^k_j}}}\right)^{\frac{\mathbbm{v}}{h_w}}\\
		=& A^{\frac{\mathtt{s_{ID^*_u}}\cdot \mathbbm{v}}{h_w\prod_{\forall r^k_j\in \mathbb{R}_{r^k_x}}\alpha_{r^k_j}}}\cdot g^{\frac{\vartheta_k\cdot \mathbbm{v}}{h_w\prod_{\forall r^k_j\in \mathbb{R}_{r^k_x}}\alpha_{r^k_j}}}\\
		tr^2_{r^k_x}=& \left(g^{\frac{a\cdot b\cdot \mathtt{s_{ID^*_u}}+ b\cdot \vartheta_k}{b\cdot \alpha_{r^k_x}}}\right)^{\frac{\mathbbm{v}}{h_w}}= A^{\frac{\mathtt{s_{ID^*_u}}\cdot \mathbbm{v}}{h_w\cdot \alpha_{r^k_x}}}\cdot g^{\frac{\vartheta_k\cdot \mathbbm{v}}{h_w\cdot \alpha_{r^k_x}}}
		\end{align}
		\par 
		Finally, simulator $\mathcal{B}$ sends trapdoor $\mathtt{Trap}$ to the adversary $\mathcal{A}_2$.
		\par  
		\textbf{\textsc{Challenge}} When adversary $\mathcal{A}_2$ decides  that \textbf{\textsc{Phase 1}} is over, it submits two equal length keywords ${w_0}$ and ${w_1}$ to the simulator $\mathcal{B}$. Simulator $\mathcal{B}$ flips a random binary coin $\omega$ and encrypts $w_\omega$ with the challenged role set $\Gamma^*$.
		\par 
		Simulator $\mathcal{B}$ first computes $h_{w_\omega}= H_1(w_\omega)$. It then chooses a random element $\mathtt{K}\in \mathbb{G}_T$ and five polynomials $q_1(x), q_2(x), q_3(x), q_4(x)$ and $q_5(x)$ of degree $2, |\Gamma^*_{\Phi}|, |\Gamma^*_{\Phi}|, |\Gamma^*|$ and $|\Gamma^*|$ respectively as follows:
		
		\begin{itemize}
			\item $q_1(x)$: Simulator $\mathcal{B}$ implicitly sets $q_1(0)= c$ and randomly chooses the rest of the points to define the polynomial $q_1(x)$ completely. Note that $q_1(1)$ and $q_1(2)$ implicitly represent $d_i$ and $d_j$ of our original scheme respectively.
			\item $q_2(x)$: Simulator $\mathcal{B}$ sets $q_2(0)= q_1(1)$ and randomly chooses the rest of the  points to define $q_2(x)$ completely.
			\item $q_3(x)$: Simulator $\mathcal{B}$ sets $q_3(0)= q_1(2)$ and randomly chooses the rest of the points to defined $q_3(x)$ completely.
			\item $q_4(x)$: Simulator $\mathcal{B}$ sets $q_4(0)= q_1(1)$ and randomly chooses the rest of the points to define $q_4(x)$ completely.
			\item $q_5(x)$: Simulator $\mathcal{B}$ sets $q_5(0)= q_1(2)$ and randomly chooses the rest of the points to define $q_5(x)$ completely.
		\end{itemize}
		\par 
		Now, simulator $\mathcal{B}$ computes a challenged ciphertext $\mathbb{CT}_\omega=\big<C_1, C_2, C_3, \{C_{4k}, C'_{4k}\}_{\forall k\in \Gamma^*_{\Phi}}, \{C_{r^k_i}, C'_{r^k_i}\}_{\forall r^k_i\in \Gamma^*}\big>$, where
		\begin{align}
		C_1=& \mathtt{K}\cdot Z\\
		C_2=& g^{b\cdot \zeta_k\cdot q_1(2)}= B^{\zeta_k\cdot q_1(2)}\\
		C_3=& g^{b\cdot \varrho_k\cdot H_2\big(g^{\frac{a\cdot b\cdot h_{id^*_c}}{b\cdot \varrho_k}}\big)\cdot q_1(2)}= B^{\varrho_k\cdot H_2\big(A^{\frac{h_{id^*_c}}{\varrho_k}}\big)\cdot q_1(2)}\\
		C_{4k}=& \Big\{g^{b\cdot \vartheta_k\cdot H_2\big(g^{\frac{a\cdot b\cdot h_{id^*_c}}{b\cdot \varrho_k}}\big)\cdot q_2(i)}\\
		=& B^{\vartheta_k\cdot H_2\big(A^{\frac{h_{id^*_c}}{\varrho_k}}\big)\cdot q_2(i)}\Big\}, \text{ $1\le i\le |\Gamma^*_{\Phi}|$}\\
		C'_{4k}=& \Big\{g^{b\cdot \vartheta_k\cdot H_2\big(g^{\frac{a\cdot b\cdot h_{id^*_c}}{b\cdot \varrho_k}}\big)\cdot q_3(i)}\\
		=& B^{\vartheta_k\cdot H_2\big(A^{\frac{h_{id^*_c}}{\varrho_k}}\big)\cdot q_3(i)}\Big\}, \text{ $1\le i\le |\Gamma^*_{\Phi}|$}\\
		C_{r^k_i}=& \Big\{g^{h_{w_\omega}\cdot b\cdot q_4(i)\prod_{\forall r^k_j\in \mathbb{R}_{r^k_i}}\alpha_{r^k_j}}\\
		=& B^{h_{w_\omega}\cdot q_4(i)\prod_{\forall r^k_j\in \mathbb{R}_{r^k_i}}\alpha_{r^k_j}}\Big\}, \text{ $1\le i\le |\Gamma^*|$}\\
		C'_{r^k_i}=& \Big\{g^{h_{w_\omega}\cdot b\cdot q_5(i)\prod_{\forall r^k_j\in \mathbb{R}_{r^k_i}}\alpha_{r^k_j}}\\
		=& B^{h_{w_\omega}\cdot q_5(i)\prod_{\forall r^k_j\in \mathbb{R}_{r^k_i}}\alpha_{r^k_j}}\Big\}, \text{ $1\le i\le |\Gamma^*|$}
		\end{align}
		\par
		Similar with CPA proof \ref{cpa_proof}, the distribution of the ciphertext $\mathbb{CT}_\omega$ for $\Gamma^*$ is identical to the original scheme.  
		\par 
		\textbf{\textsc{Phase 2}} Same as \textbf{\textsc{Phase 1}}
		
		\par 
		\textbf{\textsc{Guess}} The adversary $\mathcal{A}_2$ guesses a bit $\omega'$ and sends to the simulator $\mathcal{B}$. If $\omega'=\omega$ then the adversary $\mathcal{A}_2$ wins CPA game; otherwise it fails. If $\omega'= \omega$, simulator $\mathcal{B}$ answers ``DBDH'' in the game (i.e. outputs $l= 0$); otherwise $\mathcal{B}$ answers ``random'' (i.e. outputs $l= 1$).
		\par 
		If $Z= \hat{e}(g, g)^{z}$; then $C_{1}$ is completely random from the view of the adversary $\mathcal{A}_2$. So, the received ciphertext $\mathbb{CT}_\omega$ is not compliant to the game (i.e. invalid ciphertext). Therefore, the adversary $\mathcal{A}_2$ chooses $\omega'$ randomly. Hence, the probability of the adversary $\mathcal{A}_2$ for outputting $\omega'= \omega$ is $\frac{1}{2}$.
		\par 
		
		If $Z= \hat{e}(g, g)^{abc}$, then adversary $\mathcal{A}_2$ receives a valid ciphertext. The adversary $\mathcal{A}_2$ wins the CPA game with non-negligible advantage $\epsilon$ (according to the Theorem \ref{theo2}). As such, the probability of outputting $\omega'= \omega$ for the adversary $\mathcal{A}_2$ is $\frac{1}{2}+ \epsilon$, where probability $\epsilon$ is for guessing that the received ciphertext is valid and probability $\frac{1}{2}$ is for guessing whether the valid encrypted message $C_{1}$ is related to ${w_0}$ or ${w_1}$.
		\par 
		Therefore, the overall advantage $Adv^{IND-CKA}_{\mathcal{A}_2}$ of the simulator $\mathcal{B}$ is $\frac{1}{2}(\frac{1}{2}+ \epsilon+ \frac{1}{2})- \frac{1}{2}= \frac{\epsilon}{2}$.
	\end{proof}
	\begin{table}[!t]
		\tabcolsep 2.0pt
		\centering
		\caption{NOTATIONS}
		\begin{tabular}{p{1.5cm}p{7cm}}
			\hline
			Notation  & Description
			\\[0.5ex]    \hline
			$|\Gamma|$ & Total number of roles associated with a ciphertext\\
			$|\Gamma_\Phi|$ & Total number of SAs associated with $\Gamma$ (i.e., ciphertext)\\
			$|\mathbb{S}_{\mathtt{ID_u}}|$ & Total number of roles associated with a trapdoor\\
			$n_c$ & Total number of ciphertext associated with a revoked role \\
			$n_u$ & Total number users associated with a revoked role \\
			$n_s$ & Total number SA associated with a user
		\end{tabular}
		\label{notation2}
	\end{table}
	\begin{table*}[t]
		\centering
		\caption{Functionality Comparison}
		\label{functionality}
		\begin{tabular}{|c|c|c|c|c|c|c|c|}
			\hline
			\multicolumn{1}{|l|}{}  & \begin{tabular}[c]{@{}c@{}}Authorized Keyword\\ Search\end{tabular}&Authentication & Replay Attack & \begin{tabular}[c]{@{}c@{}}Conjunctive Keyword\\ Search\end{tabular} & Revocation & Decryption & Technique\\ \hline
			\cite{Sun2016}  & \checkmark& \ding{55} & \ding{55} & \checkmark & \checkmark & \ding{55} & ABE\\ \hline
			{\cite{Hu2017}}&\checkmark&\ding{55}  & \ding{55} & \ding{55} &  \ding{55}&  \ding{55}& ABE\\ \hline
			\cite{Miao2017} & \checkmark&\ding{55} & \ding{55} &  \checkmark&  \checkmark& \checkmark & ABE\\ \hline
			\cite{Chaudhari2019} &\checkmark&\ding{55}&\ding{55}&\ding{55}&\ding{55}& \ding{55}& ABE\\ \hline
			Proposed scheme& \checkmark&\checkmark &\checkmark  &\checkmark  &\checkmark  & \checkmark &  RBE\\ \hline
		\end{tabular}
	\end{table*}
	\begin{table}[!t]
		\centering
		\caption{Evaluation of the Computation Overhead }
		\label{performance}
		\begin{tabular}{|c|l|l|}
			\hline
			\multicolumn{2}{|c|}{Operations} &    \multicolumn{1}{c|}{Computation Complexity} \\ \hline
			\multicolumn{2}{|c|}{Data Encryption} &   $(4+ |\Gamma|+ |\Gamma_\Phi|)Exp_{\mathbb{G}_1}+ Exp_{\mathbb{G}_T}$   \\ \hline
			\multicolumn{2}{|c|}{Trapdoor Generation} &  $(3+ 2|\mathbb{S}_{\mathtt{ID_u}}|)Exp_{\mathbb{G}_1}$   \\ \hline
			\multirow{3}{*}{Data Search} & Authentication &  $(2+ |\Gamma_{\Phi}|)Exp_{\mathbb{G}_1}+ 3T_p$   \\ \cline{2-3} 
			& KeySearch &  $<(|\Gamma|+ 1)Exp_{\mathbb{G}_1}+ (2+ |\Gamma|)T_p$  \\ \cline{2-3} 
			& PartialDec &  $<(|\Gamma|+ |\Gamma_{\Phi}|)Exp_{\mathbb{G}_1}+ (1+ |\Gamma|)T_p$  \\ \hline
			\multicolumn{2}{|c|}{Decryption} & $Exp_{\mathbb{G}_T}$     \\ \hline
			\multicolumn{2}{|c|}{Revocation} &   $<(1+ 2n_c+ 2n_u)Exp_{\mathbb{G}_1}$  \\ \hline
		\end{tabular}
	\end{table}
	\begin{table}[!t]
		\centering
		\caption{Evaluation of the Storage and Communication Overhead}
		\label{performance2}
		\begin{tabular}{|c|l|l|}
			\hline
			\multicolumn{2}{|c|}{Items} &    \multicolumn{1}{c|}{Overhead} \\ \hline
			\multicolumn{2}{|c|}{Ciphertext} & $(4+ 2|\Gamma|)|\mathbb{G}_1|+ |\mathbb{G}_T|$    \\ \hline
			\multicolumn{2}{|c|}{Secret key} &   $(1+ n_s)|\mathbb{Z}_q^*|+ 2|\mathbb{S}_{\mathtt{ID_u}}||\mathbb{G}_1|$  \\ \hline
			\multicolumn{2}{|c|}{Trapdoor} &  $|\mathbb{Z}_q^*|+ (3+ 2|\Gamma|)|\mathbb{G}_1|$   \\ \hline
		\end{tabular}
	\end{table}
	\begin{table}[!t]
		\caption{Computation Time (in Milliseconds) of Elementary Cryptographic Operations}
		\begin{tabular}{|p{.5cm}|l|l|p{.7cm}|l|c|p{.5cm}|}
			\hline
			\multirow{2}{*}{} & \multicolumn{2}{p{1.5cm}|}{Exponentiation} & \multirow{2}{*}{Pairing} & \multicolumn{2}{p{2.4cm}|}{Group multiplication} & \multirow{2}{*}{Hash} \\ \cline{2-3} \cline{5-6} 
			& \multicolumn{1}{c|}{$\mathbb{G}_1$} & \multicolumn{1}{c|}{$\mathbb{G}_T$} &  & \multicolumn{1}{c|}{$\mathbb{G}_1$} & \multicolumn{1}{c|}{$\mathbb{G}_T$} &  \\ \hline
			\multicolumn{1}{|c|}{\begin{tabular}[c]{@{}c@{}}Commodity \\ Laptop\end{tabular}} & $2.062$ & $0.126$ & $1.292$ & $0.008$ & $0.002$ & $0.003$ \\ \hline
			\multicolumn{1}{|c|}{Workstation} & $1.153$ & $0.091$ & $0.645$ & $0.005$ & $0.001$ & $0.002$ \\ \hline
		\end{tabular}
		\label{time_comp}
	\end{table}
	\begin{figure}[t]
		\centering
		\scalebox{3}{\includegraphics[width=2.5cm, height=2.2cm]{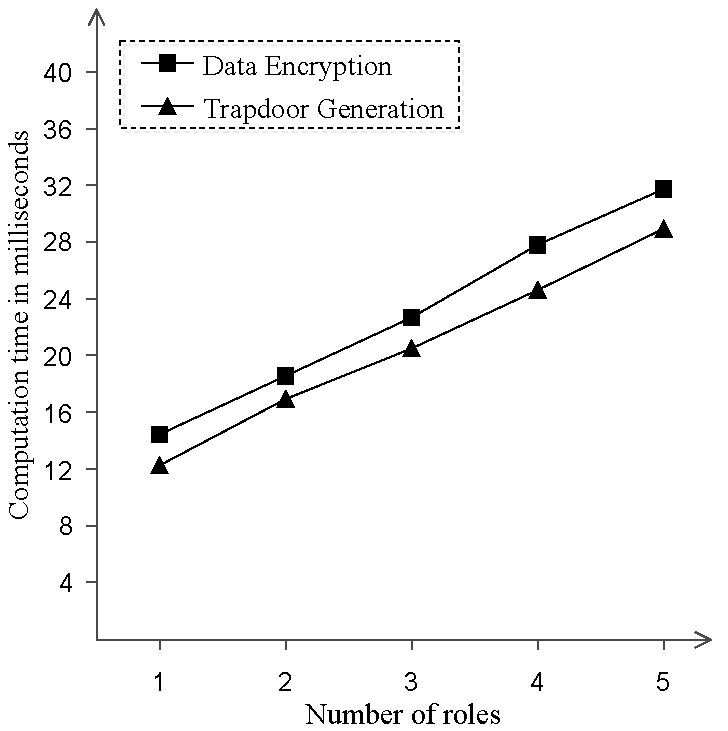}} 
		\caption{Computation Time of \emph{Data Encryption} and \emph{Trapdoor Generation} Phases}
		\label{Encryption_TrapGen}
	\end{figure}
	\begin{figure}[!t]
		\centering
		\scalebox{3}{\includegraphics[width=2.5cm, height=2.2cm]{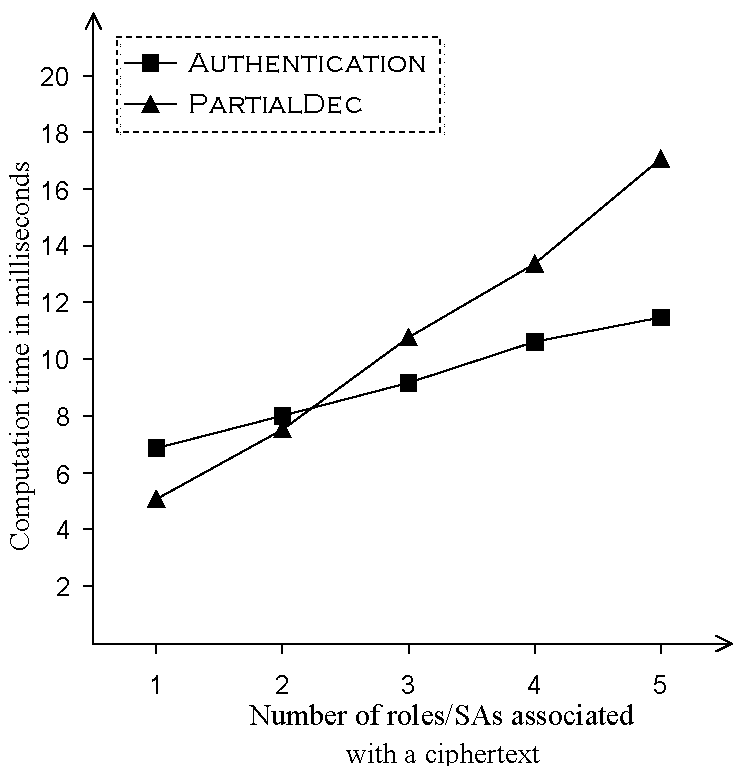}} 
		\caption{Computation Time of \textsc{Authentication} and \textsc{PartialDec} Algorithms}
		\label{Auth_PartialDec}
	\end{figure}
	\begin{figure}[!t]
		\centering
		\scalebox{3}{\includegraphics[width=2.5cm, height=2.2cm]{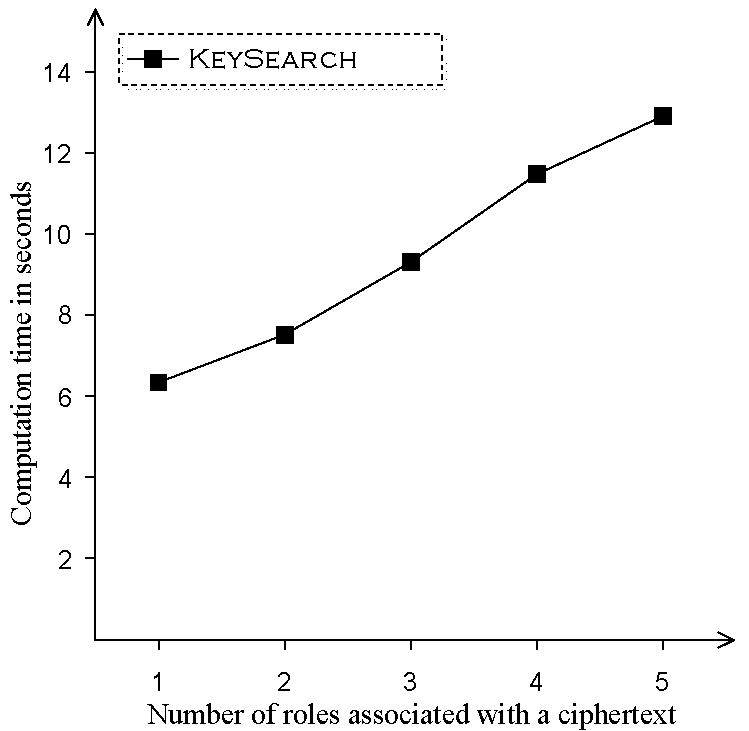}} 
		\caption{Computation Time of \textsc{KeySearch} Algorithm for 1000 Ciphertexts}
		\label{KeySeach}
	\end{figure}
	
	\subsection{Performance Analysis}
	\label{performance_analysis}
	This section evaluates functionality, computation, storage and communication overhead of our proposed scheme. The computational overhead is shown in terms of number of pairing ($T_p$) and group exponentiation operations ($Exp_{\mathbb{G}_1}$ and $Exp_{\mathbb{G}_T}$). We do not consider the other cryptographic operations such as hash and group element multiplication operations, as these operations take much less computation time compared with the pairing and group exponentiation operations (details can be seen in the Table \ref{time_comp}). The storage and communication overheads are shown in terms of group element size $|\mathbb{Z}_q^*|, |\mathbb{G}_1|$ and $|\mathbb{G}_T|$. We use PBC library \cite{pbc} which runs over GMP library \cite{gmp} for the implementation purpose. Type A elliptic curve of $160$-bit group order embedding degree $2$ is used for the implementation. The chosen curve provides an equivalent of $1024$-bit discrete log security. The elementary cryptographic operations that are performed by the owners and users are implemented using a commodity laptop Computer with Ubuntu 17.10 (64-bit) operating system and having 2.4GHz Core i3 processor with 4GB memory. The elementary cryptographic operations that are performed by public cloud is implemented using a workstation with Ubuntu 17.10 (64-bit) operating system and having 3.5 GHz Intel(R) Xeon(R) CPU E5-2637 v4 processor with 16 GB memory. Table \ref{time_comp} shows the time required to perform each cryptographic operations. During the implementation, we consider that the number of SAs associated with a RBAC access policy is equal to the number of roles associated with a ciphertext, i.e., $|\Gamma_\phi|= |\Gamma|$. It is to be noted that all the implementation results are the mean of 50 trials. The notations used in the rest of this paper are shown in Table \ref{notation2}.

	\par
	Table \ref{functionality} shows the functionality comparison of some notable ABE based keyword search schemes \cite{Sun2016, Hu2017, Miao2017, Chaudhari2019} with our proposed scheme. From the Table \ref{functionality}, it can be observed that all the ABE based schemes \cite{Sun2016, Hu2017, Miao2017, Chaudhari2019} including our proposed scheme provide authorized keyword search functionality, as the owner can embed access policies of his/her choice on the encrypted data itself. However, unlike our proposed scheme, none of the schemes in \cite{Sun2016, Hu2017, Miao2017, Chaudhari2019} address the user authentication problem, which allows the public cloud to authenticate the user before performing computationally expensive keyword search operations. As such, \cite{Sun2016, Hu2017, Miao2017, Chaudhari2019} rely on some existing authentication mechanisms. Also, unlike  \cite{Sun2016, Hu2017, Miao2017, Chaudhari2019}, our proposed scheme can prevent the replay attacks even if the trapdoors are exposed to the adversaries. In \cite{Sun2016, Hu2017, Miao2017, Chaudhari2019}, if an adversary gains access to a valid trapdoor, the adversary can re-use the trapdoor using a fresh random number. Further, our proposed scheme and \cite{Sun2016, Miao2017} support conjunctive keyword search and user revocation, while \cite{Hu2017, Chaudhari2019} do not support. Moreover, our proposed scheme and \cite{Miao2017} support both the keyword search and decryption functionalities; while  \cite{Sun2016, Hu2017, Chaudhari2019} support only the keyword search functionality. Furthermore, \cite{Sun2016, Hu2017, Miao2017, Chaudhari2019} are designed using ABE technique; while our proposed scheme is designed using RBE technique, which enables it to support the role hierarchy property. Thus, it makes our proposed scheme more suitable for the real world organizations/enterprises. Therefore, it can be observed that our proposed scheme supports more functionalities compared with the other notable works \cite{Sun2016, Hu2017, Miao2017, Chaudhari2019}. 
	\par 
	Table \ref{performance} shows the computation overhead of our proposed scheme\footnote{We do not consider \cite{Sun2016, Hu2017, Miao2017, Chaudhari2019} for further comparison, as they are based on ABE; whereas our proposed scheme is based on RBE.}. The computation cost is shown in asymptotic upper bound in the worst cases. In Table \ref{performance}, we consider the most frequently operated phases, e.g., \emph{Data Encryption}, \emph{Trapdoor Generation}, \emph{Data Search}, \emph{Decryption}, and \emph{Revocation}. 
	\paragraph{Data Encryption}
	Owner encrypts the plaintext data and the associated keywords in the \emph{Data Encryption} phase, which requires $(4+ 2|\Gamma|)$ group exponentiation operations on $\mathbb{G}_1$ and one exponentiation operation on $\mathbb{G}_T$. It can be observed that the encryption cost mainly depends on the number of roles associated with a ciphertext (i.e., associated with the chosen RBAC access policy). This can also be seen from the Figure \ref{Encryption_TrapGen}. It can be observed that approximately $31$ milliseconds are required to generate a ciphertext associated with $5$ roles and $5$ SAs. It is to be noted that, the encryption operation is performed by the owner only once for a particular data.

	\paragraph{Trapdoor Generation} 
	A user needs to perform $(3+ 2|\mathbb{S}_{\mathtt{ID_u}}|)$ group exponentiation operations on $\mathbb{G}_1$ to compute a trapdoor. It can be observed that the cost for the generation of a trapdoor depends on the number of roles associated with the user (i.e., associated with the trapdoor). Figure \ref{Encryption_TrapGen}, shows the experimental results of the \emph{Trapdoor Generation} phase, which demonstrates that our proposed scheme incurs less computation overhead on the user side. It takes approximately $29$ milliseconds to generate a trapdoor having $5$ roles.

	\paragraph{Data Search}
	In the \emph{Data Search} phase, the public cloud first authenticates the user which requires $2+ |\Gamma_{\Phi}|$ group exponentiation operations on $\mathbb{G}_1$ and three pairing operations. It can be observed that the cost of the user authentication operation (i.e., \textsc{Authentication} algorithm) depends on the number of SAs associated with the RBAC access policy of a ciphertext. Figure \ref{Auth_PartialDec} shows the computation time of \textsc{Authentication} algorithm with respect to the number of SAs. It is to be noted that the \textsc{Authentication} algorithm is performed only once per user request. After successful authentication of the user, the public cloud computes at most $|\Gamma|+ 1$ group exponentiation operations on $\mathbb{G}_1$, and $2+ |\Gamma|$ pairing operations to complete the \textsc{KeySearch} algorithm for the keyword search. It can be observed that the cost of the \textsc{KeySearch} algorithm depends on the number of roles associated with the ciphertext, which can also be seen from the Figure \ref{KeySeach}. Finally, the public cloud computes at most $|\Gamma|+ |\Gamma_{\Phi}|$ group exponentiation operations and $1+ |\Gamma|$ pairing operations to compute the \textsc{PartialDec} algorithm. It can be observed that the cost to perform the \textsc{PartialDec} algorithm depends on the number of roles and the number of SAs associated with the RBAC access policy. The computation time of \textsc{PartialDec} algorithm is shown in the Figure \ref{Auth_PartialDec}. It is to be noted that the \textsc{PartialDec} algorithm is performed for each ciphertext received from the \textsc{KeySearch} algorithm.
	

	
	\paragraph{Decryption}
	As most of the computationally expensive cryptographic operations are outsourced to the public cloud, a user requires only one group exponentiation operation on $\mathbb{G}_T$ to decrypt a ciphertext. It is to be noted that the time required to perform one group exponentiation operation on $\mathbb{G}_T$ is $0.126$ milliseconds in a commodity laptop Computer. Hence, the decryption cost in our proposed scheme is considerably less. Thus, our proposed scheme is also suitable for an environment such as IoT, where the end-users have limited computing resources. 
	
	\paragraph{Revocation}
	The complete user revocation operation takes a minimal overhead in the system, as the SA can revoke the user simply by revoking (or removing) his/her public key (from the public bulletin board). On the other hand, the SA requires at most $1+ 2n_c+ 2n_u$ group exponentiation operations on $\mathbb{G}_1$ to revoke a role from a user. As the SA needs to re-encrypt all the ciphertexts and update role-keys of all the users related with the revoked roles, the cost of the role-level revocation depends mainly on the number of ciphertext and users associated with the revoked roles. 

	\subsubsection{Storage and Communication Overhead Comparison}
	Table \ref{performance2} shows the storage and communication overhead of our proposed scheme. For the evaluation purpose, the ciphertext size, size of the secret keys possessed by a user, and the trapdoor size are considered. From Table \ref{performance2}, it can be observed that the ciphertext size mainly depends on the number of roles associated with the ciphertext. For each role $r^k_i$, the owner computes two ciphertext components $C_{r^k_i}$ and $C'_{r^k_i}$. Hence, the ciphertext size linearly increases with the roles associated with a ciphertext. 
	\par 
	A user keeps a private key $\mathtt{Priv^{k}_{ID_u}}$ for each organization, and the user also keeps a common private key $\mathtt{Priv_{ID_u}}$ for all the organizations. Moreover, the user keeps two role-keys for each role he/she possessed. Thus, the size of the secret key possessed by a user mainly depends on the number of SAs (i.e., number of organizations) and the number of roles associated with that user. Similarly, trapdoor size linearly increases with the roles associated with the trapdoor. The user computes two trapdoor components $tr^1_{r^k_x}$ and $tr^2_{r^k_x}$ for each role $r^k_x$ associated with the trapdoor.
	
	\section{Conclusion}
	\label{conclusion}
	This paper has proposed a novel authorized keyword search mechanism with efficient decryption using the RBE technique for a cloud environment, where multiple organizations can outsource their sensitive data. The proposed scheme enables the owners to define and enforce RBAC access policies on the encrypted data, thereby avoiding reducing the dependency on the service provider. It also enables the public cloud to authenticate the users first before performing computationally expensive search operations, which reduces overhead on the system. In addition, the proposed scheme helps to prevent replay attacks. Conjunctive keyword search is supported without introducing any significant overhead into the system. Further, the complete and role-level user revocation mechanisms are supported for revoking access privileges of the users in both organization level and role level respectively. Moreover, an outsourced decryption mechanism is introduced in the proposed scheme to reduce decryption processing cost at the end-user side, which makes it suitable for resource constrained environment. Furthermore, we have formally proved that the proposed scheme provides provable security against Chosen Plaintext and Chosen Keyword Attacks. Our performance analysis shows that the proposed scheme is suitable for real-world applications in terms of computation, communication and storage overhead.
	\par 
	This paper has introduced a new direction in designing a searchable encryption mechanism using the RBE technique. Further works include improving the efficiency of role-level revocation of RBE based keyword search schemes as well as for dynamic addition (removal) of roles into (from) a role hierarchy.  
	\ifCLASSOPTIONcompsoc
	\section*{Acknowledgements}
	\else
	\section*{Acknowledgement}
	
	\fi

	This paper is supported in part by European Union's Horizon 2020 research and innovation programme under the grant agreement No $830892$, project SPARTA.

	\ifCLASSOPTIONcaptionsoff
	\newpage
	\fi
	\bibliographystyle{unsrt}

\end{document}